\documentclass[runningheads]{llncs}
\pdfoutput=1
\usepackage[T1]{fontenc}
\usepackage{graphicx}

\usepackage{pdfsync}   % To jump to source: Shift-Cmd click in pdf 
\usepackage{amsmath}
\usepackage{color}
\usepackage{xspace}
\usepackage{hyperref}
\usepackage{url}
\usepackage{subcaption}
\usepackage{listings}
\usepackage{lmodern}
\usepackage{calc}
\usepackage{ifthen}
\usepackage[iso,english]{isodate}
\usepackage{marginnote}
\usepackage[scaled=0.88]{helvet}
\usepackage{pgf,tikz}
\usepackage{amssymb}
\usepackage{lipsum}
\usepackage{cite}
\usepackage[smaller,nolist,nohyperlinks]{acronym}
\usepackage{sidecap}
\usepackage{wrapfig}
\usepackage{rotating}
\usepackage[T1]{fontenc} % allow pipe in text (fixes other enc issues as well...)
\usepackage{romannum}
\usepackage{pdfpages}
\usepackage{csquotes} % allow enquote for correct quotes
\usepackage{mdframed}
\usepackage{framed}
\usepackage[multiple]{footmisc}
\usepackage{enumitem}
\usetikzlibrary{arrows,automata,positioning}

\usepackage{multirow}
\usepackage{chngcntr}

\definecolor{implicant}{RGB}{33, 166, 57}
\definecolor{implicand}{RGB}{217, 92, 20}
\newcommand{\False}{\mathsf{false}}
\newcommand{\True}{\mathsf{true}}

\DeclareGraphicsExtensions{.pdf,.jpeg,.png}
\graphicspath{{images/}}

\newboolean{anonymous} % anonymous version
\setboolean{anonymous}{true}
\newcommand{\onlyan}[1]{\ifthenelse{\boolean{anonymous}}{#1}{}}

\newcommand{\onlynonan}[1]{\ifthenelse{\boolean{anonymous}}{}{#1}}

\newcommand{\annonan}[2]{\ifthenelse{\boolean{anonymous}}{#1}{#2}}

\definecolor{keyword}{RGB}{0,0,255}
\definecolor{comment}{RGB}{0,128,0}
\definecolor{number}{RGB}{255,255,255}
\definecolor{string}{RGB}{163,21,21}

\newcommand{\eg}{e.\,g.\xspace}

\newcommand{\etal}{et~al.\xspace}

\newboolean{authorcomments}
\setboolean{authorcomments}{true}
\newcommand{\bals}[1]{\ifthenelse{\boolean{authorcomments}}{{\color{blue}\textit{[#1 -als-]}}\xspace}{}}
\newcommand{\rals}[1]{\ifthenelse{\boolean{authorcomments}}{{\color{red}\textit{[#1 -als-]}}\xspace}{}}
\newcommand{\gals}[1]{\ifthenelse{\boolean{authorcomments}}{{\color[rgb]{0,0.6,0}\textit{[#1 -als-]}}\xspace}{}}
\newcommand{\bssm}[1]{\ifthenelse{\boolean{authorcomments}}{{\color{blue}\textit{[#1 -ssm-]}}\xspace}{}}
\newcommand{\rssm}[1]{\ifthenelse{\boolean{authorcomments}}{{\color{red}\textit{[#1 -ssm-]}}\xspace}{}}
\newcommand{\gssm}[1]{\ifthenelse{\boolean{authorcomments}}{{\color[rgb]{0,0.6,0}\textit{[#1 -ssm-]}}\xspace}{}}
\newcommand{\brvh}[1]{\ifthenelse{\boolean{authorcomments}}{{\color{blue}\textit{[#1 -rvh-]}}\xspace}{}}
\newcommand{\rrvh}[1]{\ifthenelse{\boolean{authorcomments}}{{\color{red}\textit{[#1 -rvh-]}}\xspace}{}}
\newcommand{\grvh}[1]{\ifthenelse{\boolean{authorcomments}}{{\color[rgb]{0,0.6,0}\textit{[#1 -rvh-]}}\xspace}{}}
\newcommand{\bmm}[1]{\ifthenelse{\boolean{authorcomments}}{{\color{blue}\textit{[#1 -mm-]}}\xspace}{}}
\newcommand{\rmm}[1]{\ifthenelse{\boolean{authorcomments}}{{\color{red}\textit{[#1 -mm-]}}\xspace}{}}
\newcommand{\gmm}[1]{\ifthenelse{\boolean{authorcomments}}{{\color[rgb]{0,0.6,0}\textit{[#1 -mm-]}}\xspace}{}}
\newcommand{\rmac}[1]{\ifthenelse{\boolean{authorcomments}}{{\color{red}\textit{[#1 -mac-]}}\xspace}{}}
\newcommand{\bmac}[1]{\ifthenelse{\boolean{authorcomments}}{{\color{blue}\textit{[#1 -mac-]}}\xspace}{}}
\newcommand{\mmac}[1]{\ifthenelse{\boolean{authorcomments}}{{\color[rgb]{0.55, 0.0, 0.55}\textit{[#1 -mac-]}}\xspace}{}}
\newcommand{\bjep}[1]{\ifthenelse{\boolean{authorcomments}}{{\color{blue}\textit{[#1 -jep-]}}\xspace}{}}
\newcommand{\rjep}[1]{\ifthenelse{\boolean{authorcomments}}{{\color{red}\textit{[#1 -jep-]}}\xspace}{}}
\newcommand{\gjep}[1]{\ifthenelse{\boolean{authorcomments}}{{\color[rgb]{0,0.6,0}\textit{[#1 -jep-]}}\xspace}{}}

\newboolean{long} 
\setboolean{long}{false}

\newcommand{\lvsv}[2]{\ifthenelse{\boolean{long}}{#1}{#2}}

\lstdefinelanguage{stpa}{
	keywords={
		Losses, Hazards, SystemConstraints, ControlStructure, Responsibilities,
		UCAs, ContextTable, DCAs, ControllerConstraints, LossScenarios, SafetyRequirements,
		hierarchyLevel, processModel, controlActions, feedback,
		notProviding, providing, tooEarly,/,Late, stoppedTooSoon,
		for, controlAction, type, provided, contexts
	},
	frame=tb,
	showstringspaces=false,
	columns=flexible,
	basicstyle={\small\ttfamily},
	keywordstyle=\color{blue}\ttfamily\bfseries,  %\color{keyword}\bfseries,
	commentstyle=\color{comment},
	stringstyle=\color{string},
	ndkeywordstyle=\color{keyword}\bfseries,
	breaklines=true,
	breakatwhitespace=true,
	tabsize=4,
	captionpos=b,
	numbers=left,
	stepnumber=1,
	numberstyle=\tiny,
	numbersep=5pt,
	morestring=[b]",
	comment=[l]{//},
	sensitive
}

\lstdefinelanguage{SCL}
{language=C,
	morekeywords={bool, end, false, fork, fork1, input, join, join1, module,
		output, pause, par, then, tickstart, tickreturn, true},
}

\lstdefinelanguage{sctx}{
	language=C,
	morekeywords={scchart, region, state, initial, period, go, to, input, output, during, class, print, host, entry, ref},
	frame=tb,
	captionpos=b,
	breaklines=true,
	tabsize=2,
	breakindent=2mm,
	showstringspaces=false,
	stringstyle=\color{string},
	keywordstyle=\color{blue}\ttfamily\bfseries, 
	numbers=left,
	stepnumber=1,
	numberstyle=\tiny,
	numbersep=5pt,
	moredelim=**[is][\color{implicant}]{`}{`},
	moredelim=**[is][\color{implicand}]{~}{~}
}

\lstdefinelanguage{XTEND}
{language=C,
	morekeywords={def, if, null, void, for, val, var, nullOrEmpty, immutableCopy, immutableCopy,     setTypeConnector, id, setImmediate, addEffect, setTypeNormal},
	deletekeywords={const},
	moredelim=**[is][\color{red}]{@}{@},
	breaklines=true,  
	tabsize=2,
	breakindent=2mm,
	breakatwhitespace=true,
	basicstyle=\sffamily,
	columns=fullflexible,
	backgroundcolor=\color{gray!10}
}

\lstdefinelanguage{KICO}
{language=C,
	morekeywords={def, if, null, void, for, val, var, nullOrEmpty, immutableCopy, immutableCopy,     setTypeConnector, id, setImmediate, addEffect, setTypeNormal, system, pre, post, set, process, label, public},
	deletekeywords={const},
	moredelim=**[is][\color{red}]{@}{@},
	breaklines=true,  
	tabsize=2,
	breakindent=2mm,
	breakatwhitespace=true,
	basicstyle=\sffamily\scriptsize,
	columns=fullflexible,
	backgroundcolor=\color{gray!10}
}

\lstdefinelanguage{SCT}
{language=C,
	morekeywords={def, if, null, void, for, val, var, nullOrEmpty, immutableCopy, immutableCopy,     setTypeConnector, id, setImmediate, addEffect, setTypeNormal, system, pre, post, set, process, label, public, start, config, scchart, output, input, region, initial, state, go, to, connector, immediate, schedule},  
	deletekeywords={const},
	basicstyle=\footnotesize\sffamily,  
	moredelim=**[is][\color{red}]{$}{$},
	moredelim=**[is][\color{red}]{@}{@},  
	escapeinside={(*}{*)},
	breaklines=true,  
	tabsize=2,
	breakindent=2mm,
	breakatwhitespace=true,
	columns=fullflexible,
	backgroundcolor=\color{gray!10}
}

\lstdefinestyle{numbers}{
	numbers=left,
	stepnumber=1,
	numberstyle=\tiny,
	numbersep=3pt,
	xleftmargin=10pt,
}

\lstdefinestyle{nonumbers}{
	numbers=none,
	xleftmargin=3pt,
	xrightmargin=3pt,
}
\begin{acronym}[KIELER]
\acro{acc}[ACC]{Adaptive Cruise Control}
\acro{av}[AV]{Automatic Visualization}
\acro{api}[API]{Application Programming Interface}
\acro{ast}[AST]{Abstract Syntax Tree}
\acro{bmdv}[BMDV]{Federal Ministry for Digital and Transport}
\acro{cinco}[Cinco]{Cinco}
\acro{css}[CSS]{Cascading Style Sheets}
\acro{dca}[DCA]{Desired Control Action}
\acro{de}[DE]{Domain Expert}
\acro{dmtg}[dMTG]{dedicated Meta Tool Generator}
\acro{dom}[DOM]{Document Object Model}
\acro{dsl}[DSL]{Domain Specific Language}
\acro{ea}[EA]{Enterprise Architect}
\acro{ebnf}[EBNF]{Extended Backus-Naur Form}
\acro{efsm}[EFSM]{Extended \acs{fsm}}
\acro{elk}[ELK]{Eclipse Layout Kernel}
\acro{emf}[EMF]{Eclipse Modeling Framework}
\acro{fltl}[FLTL]{Fluent \ac{ltl}}
\acro{fta}[FTA]{Fault Tree Analysis}
\acro{fsm}[FSM]{Finite State Machine}
\acro{glsp}[GLSP]{Graphical Language Server Protocol}
\acro{gwt}[GWT]{Google Web Toolkit}
\acro{hci}[HCI]{Human Computer Interaction}
\acro{html}[HTML]{Hypertext Markup Language}
\acro{ide}[IDE]{Integrated Development Environment}
\acro{itdd}[ITDD]{Immersive Test-driven Development}
\acro{iur}[iur]{initialize-update-read}
\acro{j2ee}[J2EE]{Java 2 Platform, Enterprise Edition}
\acro{jdk}[JDK]{Java Development Kit}
\acro{jdt}[JDT]{Eclipse Java development tools}
\acro{json}[JSON]{JavaScript Object Notation}
\acro{json-rpc}[JSON-RPC]{\acs{json} Remote Procedure Call}
\acro{jvm}[JVM]{Java Virtual Machine}
\acro{keith}[KEITH]{\acs{kiel} Environment Integrated in Theia}
\acro{kico}[KiCo]{\acs{kieler} Compiler}
\acro{kiel}[KIEL]{Kiel Integrated Environment for Layout}
\acro{kieler}[KIELER]{Kiel Integrated Environment for Layout Eclipse Rich Client}
\acro{kiml}[KIML]{Kieler Infrastructure for Meta Layout}
\acro{klighd}[KLighD]{\acs{kieler} Lightweight Diagrams}
\acro{llvmir}[LLVM IR]{\textsc{LLVM} Intermediate Representation}
\acro{lsp}[LSP]{Language Server Protocol}
\acro{ltl}[LTL]{Linear Temporal Logic}
\acro{lts}[LTS]{Labelled Transitions Systems}
\acro{mde}[MDE]{Model-Driven Engineering}
\acro{m2mt}[M2MT]{Model-to-Model Transformation}
\acro{msc}[MSC]{Message Sequence Chart}
\acro{mte}[MTE]{Meta Tool Engine}
\acro{mtg}[MTG]{Meta Tool Generator}
\acro{mts}[MTS]{Modal Transition System}
\acro{os}[OS]{Operating System}
\acro{osm}[OSM]{Original Source Model}
\acro{osgi}[OSGi]{Open Service Gateway Initiative}
\acro{pasta}[PASTA]{Pragmatic Automated System-Theoretic Process Analysis}
\acro{pe}[PE]{Programming Expert}
\acro{pde}[PDE]{Plug-in Development Environment}
\acro{rcp}[RCP]{Rich Client Platform}
\acro{rolfer}[ROLFER]{Robotic Lifeguard For Emergency Rescue}
\acro{rssr}[RSSR]{Refined Software Safety Requirement}
\acro{ruca}[RUCA]{Refined \acs{uca}}
\acro{sahra}[SAHRA]{\acs{stpa} based Hazard and Risk Analysis}
\acro{sar}[SAR]{Search And Rescue}
\acro{sbm}[SBM]{Safe Behavior Model}
\acro{scchart}[SCChart]{Sequentially Constructive Statechart}
\acro{scg}[SCG]{Sequentially Contructive Graph}
\acro{scl}[SCL]{Sequentially Constructive Language}
\acro{slic}[SLIC]{Single-Pass Language-driven Incremental Compilation}
\acro{smv}[SMV]{Symbolic Model Verifier}
\acro{stamp}[STAMP]{System-Theoretic Accident Model and Processes}
\acro{stpa}[STPA]{System-Theoretic Process Analysis}
\acro{strm}[STRM]{Simplified Three Roles Model}
\acro{svg}[SVG]{Scalable Vector Graphic}
\acro{swt}[SWT]{Standard Widget Toolkit}
\acro{sysml}[SysML]{System Modeling Language}
\acro{te}[TE]{Tooling Expert}
\acro{tcp}[TCP]{Transmission Control Protocol}
\acro{tdd}[TDD]{Test-driven Development}
\acro{uav}[UAV]{Unmanned Aerial Vehicle}
\acro{uca}[UCA]{Unsafe Control Action}
\acro{ui}[UI]{user interface}
\acro{uml}[UML]{Unified Modeling Language}
\acro{uri}[URI]{Uniform Resource Identifier}
\acro{url}[URL]{Uniform Resource Locator}
\acro{ux}[UX]{user experience}
\acro{vscode}[VS Code]{Visual Studio Code}
\acro{xml}[XML]{Extensible Markup Language}
\acro{xstampp}[XSTAMPP]{Extensible \acs{stamp} Platform}
\acro{yang}[YANG]{Yet Another Next Generation}
\acrodef{RiCharts}{for Rich Charts}
\end{acronym}

\setboolean{anonymous}{true}

\DeclareMathOperator*{\G}{G}
\DeclareMathOperator*{\X}{X}
\DeclareMathOperator*{\U}{U}
\DeclareMathOperator*{\F}{F}
\DeclareMathOperator*{\R}{R}

\usepackage{tikz}
\usetikzlibrary{arrows, positioning, calc, arrows.meta}

\usepackage{ stmaryrd }
\usepackage{ wasysym }

\begin{document}
\counterwithout{lstlisting}{chapter}

\title{From \acs{stpa} to \aclp{sbm}\thanks{This research has been partly funded by the \ac{bmdv} within the project ``CAPTN Förde 5G''.}}

\author{Jette Petzold\orcidID{0000-0002-5559-7073} \and
Reinhard von Hanxleden\orcidID{0000-0001-5691-1215}}

\authorrunning{J. Petzold, R. von Hanxleden}

\institute{Department of Computer Science, Kiel University, Kiel, Germany\\
\email{\{jep,rvh\}@informatik.uni-kiel.de}}

\maketitle
\begin{abstract}
Model checking is a proven approach for checking whether the behavior model of a safety-critical system fulfills safety properties that are stated as \acs{ltl} formulas.
We propose rules for generating such \acs{ltl} formulas automatically based on the result of the risk analysis technique \ac{stpa}.
Additionally, we propose a synthesis of a \acl{sbm} from these generated \acs{ltl} formulas.
To also cover liveness properties in the model, we extend \ac{stpa} with \aclp{dca}.
We demonstrate our approach on an example system using \acsp{scchart} for the behavior model.
The resulting model is not necessarily complete but provides a good foundation that already covers safety and liveness properties.

\end{abstract}

\acresetall
\section{Introduction}
\label{sec:intro}

Verification is an important part of system development to ensure the safety of the system as well as its functionality.
One technique to verify a model of a system is \emph{model checking}~\cite{Clarke97}.
The system specifications are translated into \ac{ltl} formulas and a model checker determines whether the model fulfills them.
These formulas may cover safety properties as well as liveness properties.
If a formula is not fulfilled, a counterexample is generated by the model checker and the model can be adjusted and checked again.
A part of the model checking process is the translation of an \ac{ltl} formula to a Büchi automaton~\cite{GerthPVW95}.
This process is computationally expensive~\cite{GastinO01}, which is why significant research focuses on improving this translation~\cite{SomenziB00, DanieleGV99, GastinO01, BabiakKRS12, GiannakopoulouL02}.

However, creating the \ac{ltl} formulas and the model of the system in the first place is also non-trivial.
Creating models is very time-consuming, which is why many techniques exist to generate them automatically~\cite{UchitelKM01,KruegerGSB98m, DamasLL06, UchitelBC08, DippolitoBPU10}.
Additionally, if the safety of the model can already be guaranteed by the construction process, less time is needed to verify it.
The creation of the \ac{ltl} formulas can also be supported.
Generating the formulas automatically reduces the time effort and the risk of mistakes in the formulas.
If a safety property would be translated wrongly or forgotten, the verification process might overlook flaws in the model of the system.

In order to determine the safety properties for the system, \ac{stpa} \cite{LevesonT18} can be used.
\ac{stpa} is a risk analysis technique that focuses on unsafe interactions between system components and identifies more risks than traditional hazard analysis techniques \cite{Leveson16}.
In this paper we propose rules to automatically translate the resulting safety properties to \ac{ltl} formulas ensuring that no property is forgotten and reducing the time effort for creating the formulas.
Based on these formulas, we propose a synthesis of a \ac{sbm} as a statechart for the analyzed system.

\subsubsection{Contributions \& Outline}
\autoref{sec:background} introduces \ac{stpa}, the used Statecharts definition, and \ac{ltl} formulas.
\autoref{sec:relWork} reviews translation rules from \ac{stpa} to \ac{ltl} as presented by Abdulkhaleq \etal and 
the different already existing approaches for the generation of Büchi automaton from \ac{ltl} formulas.
Our main technical contributions, presented in the next three sections, are as follows:
\begin{itemize}
	\item We expand the translation of the \ac{stpa} results to \ac{ltl} formulas (\autoref{sec:ltlTranslation}).
	\item We present an approach to create an \ac{sbm} based on these formulas (\autoref{sec:sbmGeneration}).
	\item We extend \ac{stpa} with \acfp{dca}, which are needed to not only cover safety properties but liveness properties as well (\autoref{sec:dca}).
\end{itemize}

An implementation of the proposed \ac{sbm} synthesis is presented in \autoref{sec:impl}.
\autoref{sec:eva} discusses the approach and finally, \autoref{sec:conclusion} concludes the paper.

\section{Background}
\label{sec:background}

We give a short introduction of \ac{stpa} in \autoref{sec:stpa}, especially the \acp{uca} for which we propose, 
in \autoref{sec:ltlTranslation}, a translation to \ac{ltl} formulas.
\autoref{sec:statechart} presents the statechart definition used in this paper and \autoref{sec:ltl} explains the \ac{ltl} operators used in the presented formulas.

\subsection{\acs{stpa}}
\label{sec:stpa}

\ac{stpa} is a hazard analysis technique for safety critical systems  \cite{LevesonT18}.
It focuses on unsafe interactions between system components, unlike traditional techniques such as \ac{fta}, which focuses
on component failures.
The \ac{stpa} process consists of four steps \cite{LevesonT18}:

\begin{enumerate}
	\item Define the purpose of the analysis;
	\item Model the Control Structure;
	\item Identify \acfp{uca};
	\item Identify Loss Scenarios.
\end{enumerate}

In the first phase, the losses that should be prevented and hazards that lead to these losses are defined.
The control structure modeled in the second step consists of controllers, controlled processes, and possibly actuators and sensors.
For the controller \emph{control actions} are defined that can be sent to controlled processes, and the controlled processes send \emph{feedback} to the controllers.
A controller also has a process model that contains \emph{process model variables} that contain information about the controlled process, the environment, etc.

The third step is the most important one for our contributions.
Here, the control actions of the control structure are inspected.
The analyst defines \emph{\acfp{uca}} by determining in which contexts a control action causes a hazard.
The context can be stated informally by describing it, or more formally by using context tables proposed by Thomas \cite{Thomas13}.
When using context tables, the context is defined by assigning values to the process model variables.
For each control action a separate context table is created.
Each column represents a process model variable and each row a possible combination of their values called the \emph{context}.
Then, for each context the analyst can determine whether the control action is hazardous for any \ac{uca} type.
The basic \ac{uca} types are \textsc{provided} and \textsc{not-provided}.
The first one means that providing the control action leads to a hazard, while the second type states that not providing the control action leads to a hazard.
For contexts in which the timing is relevant the types \textsc{too-late} and \textsc{too-early} are used,
and for continuous control actions the types \textsc{applied-too-long} and \textsc{stopped-too-soon} must be considered as well.
Each of these types further specifies the moment in which (not) sending a control action leads to a hazard.

For these \acp{uca} \emph{controller constraints} are defined, which are the safety properties.
In the last \ac{stpa} step loss scenarios are defined that lead to hazards.

Several tools exist that support the application of \ac{stpa}, for example \ac{pasta}~\cite{PetzoldKH23}.
\ac{pasta} is a \ac{vscode} Extension that provides a \ac{dsl} with automatic visualization of the defined components and their relationships.
\ac{pasta} supports the informal definition of \acp{uca} as well as context tables.

\subsection{Statecharts}
\label{sec:statechart}

Statecharts are \acp{fsm} that are extended with hierarchy, concurrency, and communication~\cite{Harel87}.
We define a statechart $M$ based on the \ac{efsm} definition
\cite{ChengK93} as the 8-tuple $(S, I, O, D, F, U, T, s_0)$, where

\begin{itemize}
	\item $S$ is a set of \emph{states},
	\item $I$ is an $n$-dimensional space $I_1 \times \dots \times I_n$ that represents the \emph{input variables},
	\item $O$ is an $m$-dimensional space $O_1 \times \dots \times O_m$ representing the \emph{output variables},
	\item $D$ is a $p$-dimensional space $D_1 \times \dots \times D_p$ that represents the \emph{internal variables},
	\item $F$ is a set of \emph{enabling functions} $f_i$ with $f_i : D \to \{0,1\}$ that define the \emph{triggers} of the transitions,
	\item $U$ is a set of \emph{update functions} $u_i$ with $u_i : D \to D$,
	\item $T$ is a \emph{transition function} with $T : S \times F \times I \to S \times U \times O$,
	\item $s_0 \in S$ is the \emph{initial state}.
\end{itemize}

An \emph{observable trace} is a sequence of in- and outputs:
$((x_0, y_0), (x_1, y_1), \dots)$ with $x_i\in I$ input and $y_i\in O$ output in reaction $i$.
In contrast, an \emph{execution trace} also includes the states:
$((x_0, s_0, y_0), (x_1, s_1, y_1), \dots)$ with $\forall i\in \mathbb{N}: ((s_i, f_i, x_i), (s_{i+1}, u_i, y_i))\in T$  \cite{LeeS17}.
We extend this definition to include the internal variables as well:
$((x_0, z_0, s_0, y_0), (x_1, z_1, s_1, y_1), \dots)$ with $z_i$ internal variables at reaction $i$.
In the following we will just use \emph{trace} to refer to an execution trace.

\subsection{\acl{ltl}}
\label{sec:ltl}

In the following we introduce \ac{ltl} formulas for \ac{efsm} traces based on the definition by Lee and Seshia~\cite{LeeS17}.
The atomic propositions of an \ac{ltl} formula are:
\begin{itemize}
	\item $\True$;
	\item $\False$;
	\item $s$, which is $\True$ if the statechart is in state $s$;
	\item $x = v$, which is $\True$ if the variable $x$ has the value $v$.
\end{itemize}

Additionally, the boolean operators $\land, \lor, \neg$, and $\to$ can be used.
An \ac{ltl} formula $\varphi$ applies to an entire trace $t= t_0, t_1, \dots$ and holds for that trace iff $\varphi$ is true in $t_0$.
If $\varphi$ holds for all possible traces of a statechart $M$, $\varphi$ holds for $M$.
In the following let $t= t_0, t_1, \dots$ be a trace.
In order to reason about the entire trace, special temporal operators can be used:

\begin{itemize}
	\item the \emph{globally} operator $\G\varphi$ holds for $t$ when $\varphi$ holds for every suffix of $t$;
	\item the \emph{finally} operator $\F\varphi$ holds for $t$ if $\varphi$ holds for some suffix of $t$;
	\item the \emph{next} operator $\X\varphi$ holds for $t$ if $\varphi$ holds for $t_1, t_2, \dots$;
	\item the \emph{until} operator $\varphi_1 \U\varphi_2$ holds for $t$ if $\varphi_2$ holds for some suffix of $t$ and $\varphi_1$ holds until and including when $\varphi_2$ becomes $\True$;
	\item the \emph{release} operator $\varphi_1 \R \varphi_2$ holds if $\neg (\neg \varphi_1 \U \neg\varphi_2)$ holds. 
	It states that $\varphi_2$ must hold until and including when $\varphi_1 \land \varphi_2$ is $\True$ for a reaction. If $\varphi_1$ never holds, then $\varphi_2$ must hold forever.
\end{itemize}

\section{Related Work}
\label{sec:relWork}

Abdulkhaleq \etal already propose an approach for creating \ac{ltl} formulas based on \acp{uca}~\cite{AbdulkhaleqW16}.
They first define \emph{\acp{ruca}} that contain the control action, the context $CS:=\{x_i=v_i \mid x_i \text{ is a process model variable}\}$, and the type.
These \acp{ruca} are used to automatically generate \emph{\acp{rssr}}, which again are automatically translated into \ac{ltl} formulas.
The proposed \ac{ltl} formula for a context $cv := \bigwedge_{\varphi \in CS} \varphi$, 
control action $CA$, subformula $ca := sent(CA)$, and type \textsc{provided} is the following:

\begin{equation}
	\label{eq:provided}
	\G (cv \to \neg ca)
\end{equation}

It states that every time the context holds, the control action is not allowed to be sent. 
We will also use this formula for our approach.
However, we do not agree with the translations for the 
types \textsc{too-late}, \textsc{too-early}, and \textsc{not-provided} and hence provide new rules.
Additionally, we propose rules for \textsc{applied-too-long} and \textsc{stopped-too-soon}, which is not done by Abdulkhaleq \etal

To verify a model according to an \ac{ltl} formula, the negated formula is translated to a Büchi automaton, the product with the model is built, and the resulting automaton is checked for emptiness~\cite{GerthPVW95}.
Since the product automaton grows considerable in size with growing size of the Büchi automaton, 
much research focuses on optimizing the translation of an \ac{ltl} formula to a Büchi automaton to reduce the needed memory
and translation time~\cite{SomenziB00, DanieleGV99, GastinO01, BabiakKRS12, GiannakopoulouL02}.
In contrast to these works, we want to synthesize a statechart 
that can be used as the behavior model of the system.
This statechart can then be used to 
synthesize code.

For syntheses of behavior models, \ac{fltl}~\cite{GiannakopoulouM03} or scenario specifications are used.
Scenarios can be specified with \acp{msc} describing the interactions between system components and the environment.
Syntheses from \acp{msc} to \ac{lts} are presented for example by Uchitel~\etal~\cite{UchitelKM01} or Damas~\etal~\cite{DamasLL06}.
Since, \ac{lts} cannot distinguish between possible and required behavior Uchitel \etal propose a synthesis from \acp{msc} to \acp{mts}~\cite{UchitelBC08}.

Krüger \etal's approach translates \acp{msc} to statecharts~\cite{KruegerGSB98m}.
The \ac{msc} is translated to a so called \emph{\acs{msc}-automaton}.
The transitions of this automaton are then translated to intermediate states and transitions.
We use \ac{stpa} instead of \acp{msc}.
This has the advantage that the results of a risk analysis that must be done anyway can be used and no additional time is needed to create \acp{msc}.
Additionally, this eliminates the problem of \emph{implied scenarios}.
Such scenarios occur when scenarios are combined in unexpected ways resulting in unexpected system behavior not covered in the scenario specification~\cite{UchitelKM01}.

\section{\acs{stpa} to \acs{ltl}}
\label{sec:ltlTranslation}

When using context tables for specifying \acp{uca}, we implicitly have \acp{ruca} and can create \ac{ltl} formulas that prevent \acp{uca}.
The translation rules depend on the type of the \ac{uca}: \textsc{not-provided}, \textsc{provided}, \textsc{too-early}, \textsc{too-late}, \textsc{applied-too-long}, or \textsc{stopped-too-soon}.
Abdulkhaleq \etal present rules for the first four types but not for the last two.
While we agree with the rule for the \textsc{provided} type, it is not clear whether the rule for \textsc{not-provided} is correct.
This depends on how \textsc{not-provided} should be interpreted, which is discussed in \autoref{sec:not provided}.
In \autoref{sec:too-late} and \autoref{sec:too-early}, we propose more precise rules for \textsc{too-late} and \textsc{too-early}.
Additionally, we propose rules for the missing types \textsc{applied-too-long} in \autoref{sec:applied-too-long} and \textsc{stopped-too-soon} in \autoref{sec:stopped-too-soon}.

We define $PM_C:=\{PMV_1, \ldots, PMV_n\}$ as the process model of a controller $C$, 
where $PMV_i:=(x_i,\{v_1, \ldots, v_m\})$ with $1 \leq i \leq n$ is a process model variable with the name $x_i$ and possible values $v_1, \ldots, v_m$.
In the following we use an arbitrary \ac{uca} with control action $CA$
and context $CS=\{x_i=v_j \mid x_i, v_j \in PMV_i\}$ to explain the translation rules.
We define subformulas for the context variables $cv := \bigwedge_{\varphi \in CS} \varphi$,
meaning that the context in which $CA$ is hazardous is present,
and control action $ca := sent(CA)$,  meaning $CA$ is sent.
We will use a short form of the trace definition:
We use $(cv, ca)$ or $(\neg cv, \neg ca)$ instead of $(x, z, s, y)$, meaning the input and internal variables $x$ and $z$ are set according to $cv$ or $\neg cv$, respectively. 
The variables that do not occur in $cv$ can have any value,
the state $s$ can be any state, and the output $y$ must (not) contain $CA$.
\autoref{f:formulas-summary} gives an overview of the traces we prevent with the proposed formulas, as elaborated in the following.

\begin{figure*}[htp]
	\centering
	\begin{subfigure}{.56\linewidth}
		\centering
		\fbox{	
			\begin{tikzpicture}[scale=.23]
				\draw (2,2) node (trace) {$\cdots$, ($\neg cv,\neg ca$), ($cv,\neg ca$), ($\neg cv,\neg ca$), $\cdots$};
				
				\draw (-4.25,-2) node (implicant) {\textcolor{implicant}{$\neg cv\wedge \X cv$} $\checkmark$};
				\draw (9.5,-2) node (implicand) {\textcolor{implicand}{$ca \R cv$} $\text{\lightning}$};
				
				\draw[->] ($(implicant.north)-(0.75,0)$) -- ($(trace.south)-(7, 0)$);
				\draw[->] ($(implicand.north)-(0.5,0)$) -- ($(trace.south)+(7, 0)$);
			\end{tikzpicture}
		}
		\caption{\centering Control action is not provided.\newline $\G($\textcolor{implicant}{$(\neg cv \land \X cv$)}$ \to \X ($\textcolor{implicand}{$ca \R cv$}$\land \F ca))$.}
		\label{f:traces-not-provided}
	\end{subfigure}
	\hfill
	\begin{subfigure}{.4\linewidth}
		\centering
		\fbox{	
			\begin{tikzpicture}[scale=.23]
				\draw (2,2) node (trace) {$\cdots$, ($cv,ca$), $\cdots$};
				
				\draw (0,-2) node (implicant) {\textcolor{implicant}{$cv$} $\checkmark$};
				\draw (4,-2) node (implicand) {\textcolor{implicand}{$\neg ca$} $\text{\lightning}$};
				
				\draw[->] ($(implicant.north)+(2,0)$) -- ($(trace.south)$);
			\end{tikzpicture}
		}
		\caption{\centering Control action is provided.\newline $\G ($\textcolor{implicant}{$cv$}$\to $\textcolor{implicand}{$\neg ca$}$)$}
		\label{f:traces-provided}
	\end{subfigure}
	\par\medskip
	\begin{subfigure}{.47\linewidth}
		\centering
		\fbox{	
			\begin{tikzpicture}[scale=.23]
				\draw (2,2) node (trace) {$\cdots$, ($\neg cv,ca$), ($cv,ca$), $\cdots$};
				
				\draw (-3,-2) node (implicant) {\textcolor{implicant}{$\neg cv \wedge \X cv$} $\checkmark$};
				\draw (3.5,-2) node (implicand) {\textcolor{implicand}{$\neg ca$} $\text{\lightning}$};
				
				\draw[->] ($(implicant.north)+(2.5,0)$) -- ($(trace.south)-(2.5, 0)$);
			\end{tikzpicture}
		}
		\caption{\centering Control action is provided too early.\newline $\G($\textcolor{implicant}{$(\neg cv \wedge \X cv)$}$\to$\textcolor{implicand}{$\neg ca$}$)$}
		\label{f:traces-too-early}
	\end{subfigure}
	\hfill
	\begin{subfigure}{.5\linewidth}
		\centering
		\fbox{	
			\begin{tikzpicture}[scale=.23]
				\draw (2,2) node (trace) {$\cdots$, ($\neg cv,\neg ca$), ($cv,\neg ca$), $\cdots$};
				
				\draw (-1,-2) node (implicant) {\textcolor{implicant}{$\neg cv$} $\checkmark$};
				\draw (6.25,-2) node (implicand) {\textcolor{implicand}{$cv \to ca$} $\text{\lightning}$};
				
				\draw[->] ($(implicant.north)-(0.5,0)$) -- ($(trace.south)-(3.5, 0)$);
				\draw[->] ($(implicand.north)-(0.75,0)$) -- ($(trace.south)+(3.5, 0)$);
			\end{tikzpicture}
		}
		\caption{\centering Control action is provided too late.\newline $\G($\textcolor{implicant}{$\neg cv$}$\to \X$\textcolor{implicand}{$(cv \to ca)$}$)$}
		\label{f:traces-too-late}
	\end{subfigure}
	\par\medskip
	\begin{subfigure}{.47\linewidth}
		\centering
		\fbox{	
			\begin{tikzpicture}[scale=.23]
				\draw (2,2) node (trace) {$\cdots$, ($cv,ca$), $\ \ $  ($\neg cv, ca$), $\cdots$};
				
				\draw (-1.5,-2) node (implicant) {\textcolor{implicant}{$cv \wedge ca$} $\checkmark$};
				\draw (6.75,-2) node (implicand) {\textcolor{implicand}{$\neg cv \to \neg ca$} $\text{\lightning}$};
				
				\draw[->] ($(implicant.north)-(0.5,0)$) -- ($(trace.south)-(4, 0)$);
				\draw[->] ($(implicand.north)-(0.75,0)$) -- ($(trace.south)+(4, 0)$);
			\end{tikzpicture}
		}
		\caption{\centering Control action is applied too long.\newline $\G($\textcolor{implicant}{$(cv \wedge ca)$}$\to \X$\textcolor{implicand}{$(\neg cv \to \neg ca)$}$)$}
		\label{f:traces-applied-too-long}
	\end{subfigure}
	\hfill
	\begin{subfigure}{.5\linewidth}
		\centering
		\fbox{	
			\begin{tikzpicture}[scale=.23]		
				\draw (2,2) node (trace) {$\cdots$, ($cv,ca$), $\ \ $  ($cv,\neg ca$), $\cdots$};
				
				\draw (-1.5,-2) node (implicant) {\textcolor{implicant}{$cv \wedge ca$} $\checkmark$};
				\draw (6.75,-2) node (implicand) {\textcolor{implicand}{$\neg ca \to \neg cv$} $\text{\lightning}$};
				
				\draw[->] ($(implicant.north)-(0.5,0)$) -- ($(trace.south)-(4, 0)$);
				\draw[->] ($(implicand.north)-(0.75,0)$) -- ($(trace.south)+(4, 0)$);
			\end{tikzpicture}
		}
		\caption{\centering Control action is stopped too soon.\newline $\G ($\textcolor{implicant}{$(cv \wedge ca)$}$\to \X$\textcolor{implicand}{$(\neg ca \to \neg cv)$}$)$}
		\label{f:traces-stopped-too-soon}
	\end{subfigure}
	\caption{Traces for the different \ac{uca} types that are prevented by the proposed formulas. Here, we do not consider the first reaction of a trace.}
	\label{f:formulas-summary}
\end{figure*}
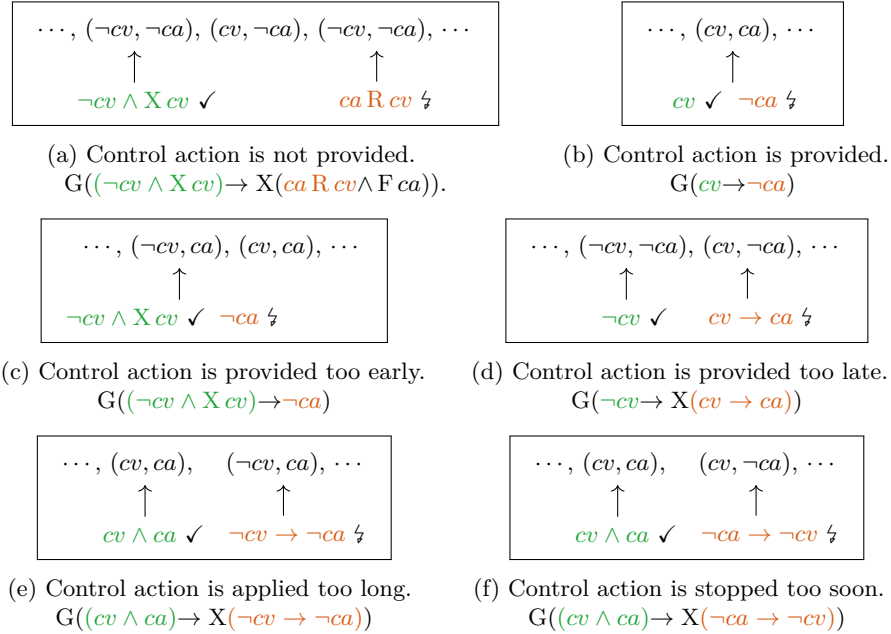

\subsection{Formula for Not-Provided}
\label{sec:not provided}

Stating that not providing a control action is hazardous in a given context can be interpreted in two ways.
On the one hand, this could mean that in every reaction where the context holds the control action must be sent to prevent a hazard.
On the other hand, it could mean that during the timespan where the context holds continuously the control action must be sent at least once.
The rule stated by Abdulkhaleq \etal~\cite{AbdulkhaleqW16} covers the first interpretation:
\[\G (cv \to ca)\]

However, this interpretation also already ensures that the control action is not sent too late and is not stopped too soon.
That is why we propose a formula for the second interpretation. 
Let 
\[\psi = cv \to (ca \R cv \land \F ca),\]
\[\chi = \G((\neg cv \land \X cv) \to \X (ca \R cv \land \F ca))\]
then we translate a \ac{uca} of type \textsc{not-provided} to the following formula:
\begin{equation}
	\label{eq:not-provided}
	\psi \land \chi
\end{equation}

In $\chi$ the implicant $\neg cv \land \X cv$ holds in the reactions directly before $cv$ changes from $\False$ to $\True$.
This means the next reaction is the first time where the context holds, and in this reaction $((ca \R cv) \land \F ca)$ should hold.
$(ca \R cv)$ ensures that when $cv$ changes to $\False$, 
the control action has to have been sent before.
Thus, a trace as shown in \autoref{f:traces-not-provided}
would evaluate to $\False$ since the context changes to $\False$ although the control action was not yet sent.
Since $(ca \R cv)$ evaluates to $\True$ when the control action is not sent as long as the context holds indefinitely, we ensure with $\F ca$ that eventually the control action will be sent.
Hence, the implicand ensures that after the context switches from $\False$ to $\True$, the control action is at least sent once while the context holds.
However, when a trace starts with $cv$ being $\True$, the implication evaluates to $\True$ regardless whether the control action is sent before $cv$ changes to $\False$ since the implicant evaluates to $\False$.
$\psi$ ensures that also in the first reaction when $cv$ is $\True$, $ca$ must be sent at least once before it changes to $\False$.

\subsection{Formula for Too-Late}
\label{sec:too-late}
Abdulkhaleq \etal propose the following formula for \acp{uca} of type \textsc{too-late}~\cite{AbdulkhaleqW16}:
\[\G((cv \to ca) \land \neg(cv \U ca))\]
This formula can only be fulfilled if the first conjunct is fulfilled, namely $cv \to ca$.
This means traces such as $(\dots, (\neg cv, \neg ca), (cv,ca), (cv, \neg ca), \dots)$ 
evaluate to $\False$ although the control action is not provided too late.
Hence, this formula covers more than just the \textsc{too-late} \ac{uca} type.
Even in cases where $cv \to ca$ holds, the complete formula evaluates to $\False$ for some trace in which the control action is not provided too late.
For example, for the trace $((cv,ca), (cv, ca), \dots)$. 
Since $cv U ca$ evaluates to $\True$ for this trace, $\neg(cv \U ca)$ evaluates to $\False$ and hence the complete formula evaluates to $\False$.

We argue that 
we only have to look at the first moment the control action should be provided.
The control action should be sent instantly when the context holds.
This leads to the following formula:
\begin{equation}
	\label{eq:too-late}
	(cv \to ca) \land \G(\neg cv \to \X(cv \to ca))
\end{equation}
The second conjunct ensures that traces such as shown in \autoref{f:traces-too-late} do not occur.
The moment the control action should be applied is when the context currently holds and in the previous reaction
did not hold.
In the formula we capture this in the following way:
The control action should be applied in the next reaction if the context currently does not hold and in the next reaction does hold.
If the control action is not applied, the formula evaluates to $\False$.
The first conjunct of \autoref{eq:too-late}, namely $cv \to ca$, just ensures that the \ac{uca} does not occur in the first reaction.
If the context already holds in the first reaction we must apply the control action immediately.

\subsection{Formula for Too-Early}
\label{sec:too-early}
For the \acp{uca} of type \textsc{too-early} Abdulkhaleq \etal propose the following formula~\cite{AbdulkhaleqW16}:
\[\G((ca \to cv) \land \neg(ca \U cv))\]
We see here the same problem as before.
The first conjunct, namely $ca \to cv$, ensures that the control action is not sent too early, but this again covers too much.
For traces such as $(\dots, (\neg cv, \neg ca), (cv,ca), (\neg cv, ca), \dots)$, where the control action is not sent too early, the formula evaluates to $\False$ because of the reaction $(\neg cv, ca)$ 
and hence the complete formula evaluates to $\False$.
Let us examine a trace where the control action is not sent too early and the subformula holds:
$((cv,ca), (\neg cv, \neg ca), \dots)$.
Since in the first reaction $cv$ and $ca$ hold, the formula $ca \U cv$ evaluates to $\True$ and hence $\neg(ca \U cv)$ to $\False$, which is why the whole formula is evaluated to $\False$.

We propose a formula that only inspects the moment where $cv$ switches from $\False$ to $\True$:
\begin{equation}
	\label{eq:too-early}
	\G((\neg cv \land \X cv) \to \neg ca)
\end{equation}

The reaction where the context holds but did not hold in the previous reaction is the first one where the control action is allowed to be sent.
Hence, we must ensure that before this reaction, the control action is not sent (see \autoref{f:traces-too-early}).
The implicant, namely $\neg cv \land \X cv$, is $\True$ if the current reaction is the last one in which $cv$ is $\False$ before it switches to $\True$.
In such a reaction the control action is not allowed to be sent, which is guaranteed by the implication.

\subsection{Formula for Applied-Too-Long}
\label{sec:applied-too-long}

For \acp{uca} of type \textsc{applied-too-long} we propose the following formula:
\begin{equation}
	\label{eq:applied-too-long}
	\G((cv \land ca) \to \X (\neg cv \to \neg ca))
\end{equation}

To ensure an action is not applied too long, we have to inspect the reactions where the control action is already applied while the context holds ($cv \land ca$).
In these reactions, we must ensure that the control action is not sent anymore at the latest when the context does not hold any longer.
Hence, for each reaction in which $cv$ and $ca$ are $\True$ ($cv \wedge ca$), we must check whether the context still holds in the next reaction ($\X(\dots)$).
If it does not, the control action must not be sent, which is guaranteed by the subformula $\neg cv \to \neg ca$.
We are not interested in the traces where the control action is stopped before the context does no longer hold.
This is only relevant for \acp{uca} of type \textsc{stopped-too-soon}.
The formula only evaluates to $\False$ for traces such as shown in \autoref{f:traces-applied-too-long}.

\subsection{Formula for Stopped-Too-Soon}
\label{sec:stopped-too-soon}

In order to guarantee that a control action is not stopped too soon, we must ensure that after it is sent the first time it is continuously sent until the context does not longer hold.
We propose a similar formula as for
\textsc{applied-too-long}:
\begin{equation}
	\label{eq:stopped-too-soon}
	\G ((cv \land ca) \to \X (\neg ca \to \neg cv))
\end{equation}

Again, we are interested in the reactions where the context already holds and the control action is applied ($cv \land ca$).
In such reactions, the formula ensures that if the control action is not sent anymore in the next reaction, then the context does not hold in the next reaction ($\X (\neg ca \to \neg cv)$).
This way traces such as shown in \autoref{f:traces-stopped-too-soon} where the control action is stopped too soon are prevented.

\section{\acs{stpa} to \acs{sbm}}
\label{sec:sbmGeneration}
The \ac{ltl} formulas generated based on the identified \acp{uca} can be used for model checking the behavior models of the software controllers.
Those \emph{\acfp{sbm}} are created manually by the software developers with the help of supporting tools such as Simulink.
Abdulkhaleq \etal~\cite{AbdulkhaleqW16} present an approach to combine the generated \ac{ltl} formulas with an \ac{sbm} to create a \ac{smv} model that can be verified using for example NuSmv \cite{CimattiCGR99}. 
This way, the safety of the model is examined.
However, creating an \ac{sbm} in the first place is time-consuming and error-prone.

We now propose an approach for the automatic generation of deterministic \acp{sbm} based on the \ac{ltl} formulas generated for the controller that should be modeled.
This way, the developer can work with an initial \ac{sbm} that already fulfills the \ac{ltl} formulas that are used for the generation.
For this automatic generation we assume that no contradicting \acp{uca} exist.
Conflicts have to be solved before applying the generation.
Additionally, we assume for simplicity that only one control action is sent at a time 
and that the initial reaction is used to set up the system such that we do not have to consider subformulas that are only checked on the initial reaction.

The first question we need to answer when generating an \ac{sbm} from \ac{stpa} is how we determine the states for the model, which is explained in \autoref{sec:sbm-gen-states}. 
Afterwards, \autoref{sec:sbm-gen-variables} introduces the determination of the variables in the \ac{sbm}.
\autoref{sec:sbm-gen-not-provided} - \autoref{sec:sbm-gen-stopped-too-soon} present a translation for each \ac{uca} type from the corresponding \ac{ltl} formula to transitions and possibly states
as well as their interaction with the other \ac{ltl} formula translations.
For this we use the previously introduced formulas.
Finally, we optimize the constructed \ac{sbm} (\autoref{sec:sbm-gen-optimization}) and prove that it fulfills the \ac{ltl} formulas except the ones for \textsc{too-early} (\autoref{sec:sbm-gen-proof}).

\subsection{States}
\label{sec:sbm-gen-states}

A straightforward approach to determine the states for the model would be to use each context and control action combination.
However, this would blow up the state space, and it is not always necessary to differentiate between states with the same control action but different context.
Thus, we start with states that only represent a control action each and an initial state that represents that no control action is sent.
Additional states are added during the translation of \ac{ltl} formulas if necessary.
Hence, for a controller $C$ we start with a statechart $M=(S, I, O, D, F, U, T, s_0)$
with $S=s_0 \cup \{ s_{CA} \mid CA \text{ is a control action of } C\}$ and $F = U = T=\emptyset$.

\subsection{Variables}
\label{sec:sbm-gen-variables}

The variables of the statechart are determined based on the process model of the controller $C$.
All process model variables are translated to internal variables, meaning we declare $D=\{x \mid (x,\_)\in PM_C\}$.
We infer the type of each input variable from the values it can be assigned.
If the values are $\True$ and $\False$, the type of the variable is \emph{boolean}, otherwise it is a \emph{number}.
The input variables are composed of the variables that are used for the possible values of the process model variables.
Hence, we set $I=\{v \mid v\in V, (\_, V) \in PM_C, v \text{ is variable} \}$.
We track the control action that is sent in an additional variable named \emph{control\-Action}.
Since the chosen control action should be sent to a system component, we declare \emph{controlAction} as an output variable: $O = \{controlAction\}$.
The type of this variable depends on the implementation: \emph{string} can be used or an \emph{enum} can be created containing all possible control actions.
Two options exist to set the value of this variable in each reaction:
Each state can define an entry action setting \emph{controlAction} to the value the state represents,
or each transition must set \emph{controlAction} according to the target state.
In the following we will use the first option, meaning we can concentrate on the triggers of the transitions.

\subsection{Not-Provided Transitions}
\label{sec:sbm-gen-not-provided}
The translation of \ac{ltl} formulas for \acp{uca} of type \textsc{not-provided} can be seen in \autoref{f:not-provided}.
We cannot automatically determine at which moment in the timespan, where the context holds, the control action should be sent.
Thus, we set the moment where the control action must be sent to the first reaction where $cv$ holds, which also covers the \textsc{too-late} type.
Since each state in our statechart represents that a specific control action is sent,
we must ensure that if the context holds, we go to the state representing the control action.
For each formula we add transitions from the states not representing $ca$ to the one representing it:
$T(s, f_{s}, i) = (s_{ca}, \emptyset, \emptyset)$ if $f_{s}(d)=1$, with  $s \neq s_{ca}$, $i\in I$, $d \in D$, and $f_{s}(cv) = 1$.

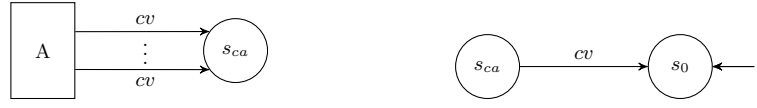
\begin{figure*}[htp]
	\centering
	\begin{subfigure}{.52\linewidth}
		\centering
			\begin{tikzpicture}[>=stealth', scale=.85, transform shape]
				\node[draw, minimum width=1cm, minimum height=1.5cm] (A) {A};
				\node[circle, draw, minimum size=1cm, right=2cm of A] (C) {$s_{ca}$};
				
				\draw[->] ($(A.east)-(0, .296)$) -- (C.216);
				\draw[->] ($(A.east)+(0, .296)$) -- (C.144);
				
				\node[left = 1cm of C.216, below] (l1) {\(cv\)};
				\node[left = 1cm of C.144, above] (l2) {\(cv\)};
				\node at ($(l1)!.5!(l2)$) {\(\strut\vdots\)};
			\end{tikzpicture}
		\caption{Translation of \acs{uca} type \textsc{not-provided}.}
		\label{f:not-provided}
	\end{subfigure}
	\hfill
	\begin{subfigure}{.45\linewidth}
		\centering
			\begin{tikzpicture}[>=stealth', scale=.85, transform shape]
				\node[circle, draw, minimum size=1cm] (C) {$s_{ca}$};
				\node[circle, draw, minimum size=1cm, right=2cm of C, initial, initial text = {}, initial distance = 0.65cm, initial right] (CC) {$s_0$};
				
				\draw[->] (C) -- node[above] {\(cv\)} (CC);
			\end{tikzpicture}
		\caption{Translation of \acs{uca} type \textsc{provided}.}
		\label{f:provided}
	\end{subfigure}

	\caption{Translation of \acs{uca} types \textsc{provided} and \textsc{not-provided} to statecharts. $A$ is the statechart without the state $s_{ca}$.}
	\label{f:translations-to-sbm}
\end{figure*}

\subsection{Too-Early Transitions}
\label{sec:sbm-gen-too-early}

\ac{ltl} formulas for the \textsc{too-early} \ac{uca} type (\autoref{eq:too-early}) cannot be translated.
In these formulas the context in the next reaction constrains the current control action. 
Since we cannot see the future, we cannot depict that.

\subsection{Too-Late Transitions}
\label{sec:sbm-gen-too-late}

The first conjunct of the \ac{ltl} formula for the \textsc{too-late} type (\autoref{eq:too-late}) is only relevant for the first reaction.
Since we assume that the initial reaction is used for setting up the system, we can ignore it.
For the second conjunct, we must remember the context in the previous reaction.
We could do that by adding new states $s_{a\_cv}$ to $S$, which represent that $cv$ holds and $a$ is sent, for each $a \neq ca$.
Then, we could add transitions from $s_{a\_cv}$ to $s_{a}$ that trigger when $cv$ does not hold.
The incoming transitions of $s_{a}$ are split such that the ones where the trigger contains $\neg cv$ remain and the ones containing $cv$ are changed such that they go to $s_{a\_cv}$ instead.
Since $s_{a}$ represents that $cv$ did not hold in the last reaction, we could add transitions from $s_{a}$ to $s_{ca}$ that trigger when $cv$ holds.
This way, we would depict the implication $\neg cv \to \X(cv \to ca)$.
However, the $s_{a\_cv}$ states would be unreachable.
The only way a transition would go from an arbitrary state to $s_{a\_cv}$ would be if a \ac{uca} of type not-providing exists for the control action $a$ and context $cv$.
This would be a contradiction to the current inspected formula that $ca$ should not be sent too late in context $cv$ and hence does not occur in a correct analysis.
In conclusion, we do not need extra states, we just need to add transitions from states not representing $ca$ to $s_{ca}$ just as done for \textsc{not-provided}.

\subsection{Provided Transitions}
\label{sec:sbm-gen-provided}

The translation for \acp{uca} of type \textsc{provided} is depicted in \autoref{f:provided}.
In order to fulfill these formulas, 
we need contrary transitions to the ones introduced for the \textsc{not-provided} \ac{uca} type.
We must ensure that if the context holds, the control action is not sent.
This is done by adding a transition from $s_{ca}$ to the initial state:
$T(s_{ca}, f_{s}, i) = (s_0, \emptyset, \emptyset)$ if $f_{s}(d)=1$, with $i\in I$, $d \in D$, and $f_{s}(cv) = 1$.
However, we must consider that for the same context or a context containing $cv$ a \ac{uca} of type \textsc{not-provided} or \textsc{too-late} may have been defined.
In the first case, we already have a transition that leaves $s_{ca}$ when $cv$ holds and we do not have to add another one to the initial state.
In the second case, we have to further specify the trigger such that it does not evaluate to $\True$ if the other transitions evaluates to $\True$.
Otherwise, the statechart would be non-deterministic.

\subsection{Applied-Too-Long Transitions}
\label{sec:sbm-gen-applied-too-long}

The formulas for the \ac{uca} type \textsc{applied-too-long} (\autoref{eq:applied-too-long}) may already be covered by the previously added transitions.
In order to check this, the outgoing transitions of $s_{ca}$ must be inspected.
If a transition is triggered for every context $c \neq cv$, the formula is already fulfilled.
Otherwise, we need to split $s_{ca}$ into two states by adding $s_{ca\_cv}$ to $S$.

The result of the translation can be seen in \autoref{f:applied-too-long}.
The original transitions to $s_{ca}$ are updated in the following way:
If previously $T(s, f_{s}, i) = (s_{ca}, \emptyset, \emptyset)$ with $f_{s}(cv) = 1$, we remove that transition and add $T(s, f'_{s}, i) = (s_{ca\_cv}, \emptyset, \emptyset)$ with $f'_s(x \land cv) = 1 \Leftrightarrow f_s(x) = 1$ and 
$T(s, f'_{s}, i) = (s_{ca}, \emptyset, \emptyset)$ with $f'_s(x \land \neg cv) = 1 \Leftrightarrow f_s(x) = 1$.
Hence, the transitions to $s_{ca}$ only trigger when $cv$ does not hold.
If $cv$ holds, the transition to $s_{ca\_cv}$ is triggered.
Additionally, we copy all outgoing transitions of $s_{ca}$ to $s_{ca\_cv}$ except the ones going to other duplicate states of $s_{ca}$ 
and add a transition between the two states that triggers when $cv$ holds:
$T(s_{ca}, f_{ca}, i) = (s_{ca\_cv}, \emptyset, \emptyset)$ if $f_{ca}(d)=1$, with $i\in I$, $d \in D$, and $f_{ca}(cv) = 1$.
Now, $s_{ca\_cv}$ represents that $cv$ holds and the control action $ca$ is sent.

In order to fulfill the formula, we must still assure that the control action is not sent anymore when the context changes.
Hence, we add another transition from  $s_{ca\_cv}$ to the initial state triggering when $cv$ does not hold:
$T(s_{ca\_cv}, f_{s}, i) = (s_0, \emptyset, \emptyset)$ if $f_{ca\_cv}(d)=1$, with  $i\in I$, $d \in D$, and $f_{s}(\neg cv) = 1$.
In order to avoid non-determinism the trigger may have to be adjusted such that it only evaluates to $\True$ when the trigger of the other transitions evaluate to $\False$.
This can be done by modifying $f_{s}$ such that it only evaluates to $1$ when $\neg cv$ holds and the triggers of all other transitions do not hold.

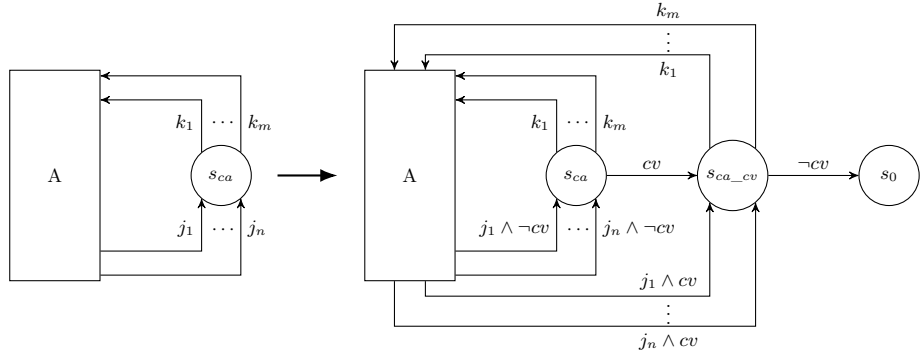
\begin{figure*}[htb]
	\centering
	\begin{tikzpicture}[>=stealth', scale=.8, transform shape]
		\node[draw, minimum width=1.5cm, minimum height=3.5cm] (A) {A};
		\node[circle, draw, minimum size=1cm, right=1.5cm of A] (C) {$s_{ca}$};
		
		\draw[->] ($(A.south east)+(0, .5)$) -| (C.230);
		\draw[->] ($(A.south east)+(0, .1)$) -| (C.310);
		\node[below = .5cm of C.230, left] (l1) {\(j_1\)};
		\node[below = .5cm of C.310, right] (l2) {\(j_n\)};
		\node at ($(l1)!.5!(l2)$) {\(\strut\cdots\)};
		
		\draw[->] (C.50) |- ($(A.north east)-(0, .1)$);
		\draw[->] (C.130) |- ($(A.north east)-(0, .5)$);
		\node[above = .5cm of C.130, left] (l3) {\(k_1\)};
		\node[above = .5cm of C.50, right] (l4) {\(k_m\)};
		\node at ($(l3)!.5!(l4)$) {\strut\(\cdots\)};
		\draw[thick, -Latex] (3.7, 0) -- +(1, 0);
		\begin{scope}[xshift=5.9cm]
			\node[draw, minimum width=1.5cm, minimum height=3.5cm] (A) {A};
			\node[circle, draw, minimum size=1cm, right=1.5cm of A] (C) {$s_{ca}$};
			\node[circle, draw, minimum size=1cm, right=1.5cm of C] (CC) {$s_{ca\_cv}$};
			\node[circle, draw, minimum size=1cm, right=1.5cm of CC] (s) {\(s_0\)};
			
			\draw[->] ($(A.south east)+(0, .5)$) -| (C.230);
			\draw[->] ($(A.south east)+(0, .1)$) -| (C.310);
			\node[below = .5cm of C.230, left] (l1) {\(j_1\land\lnot cv\)};
			\node[below = .5cm of C.310, right] (l2) {\(j_n\land\lnot cv\)};
			\node at ($(l1)!.5!(l2)$) {\(\strut\cdots\)};
			
			\draw[->] ($(A.south west)+(.5, 0)$) -- +(0, -.75) -| (CC.310);
			\draw[->] ($(A.south east)-(.5, 0)$) -- +(0, -.25) -| (CC.230);
			
			\node[yshift=-.25cm, right=4.3cm of A.south, above] (l5) {\(j_1\land cv\)};
			\node[yshift=-.75cm, right=4.3cm of A.south, below] (l6) {\(j_n\land cv\)};
			\node at ($(l5)!.5!(l6)$) {\(\strut\vdots\)};
			
			\draw[->] (C.50) |- ($(A.north east)-(0, .1)$);
			\draw[->] (C.130) |- ($(A.north east)-(0, .5)$);
			\node[above = .5cm of C.130, left] (l3) {\(k_1\)};
			\node[above = .5cm of C.50, right] (l4) {\(k_m\)};
			\node at ($(l3)!.5!(l4)$) {\strut\(\cdots\)};
			\draw[->] (C) -- node[above] {\(cv\)} (CC);
			\draw[->] (CC) -- node[above] {\(\lnot cv\)} (s);
			
			\draw[<-] ($(A.north west)+(.5, 0)$) -- +(0, .75) -| (CC.50);
			\draw[<-] ($(A.north east)-(.5, 0)$) -- +(0, .25) -| (CC.130);
			
			\node[yshift=.25cm, right=4.3cm of A.north, below] (l7) {\(k_1\)};
			\node[yshift=.75cm, right=4.3cm of A.north, above] (l8) {\(k_m\)};
			\node at ($(l7)!.5!(l8)$) {\(\strut\vdots\)};
		\end{scope}
	\end{tikzpicture}
	\caption{Translation of \acs{uca} type \textsc{applied-too-long}. $A$ is the statechart without the state $s_{ca}$.}
	\label{f:applied-too-long}
\end{figure*}

\subsection{Stopped-Too-Soon Transitions}
\label{sec:sbm-gen-stopped-too-soon}

\ac{ltl} formulas for \acp{uca} of type \textsc{stopped-too-soon} (\autoref{eq:stopped-too-soon}) are translated in a similar way (\autoref{f:stopped-too-soon}).
If in $s_{ca}$ no transition is triggered for $cv$, the formula is already fulfilled.
Otherwise, we need to split $s_{ca}$ and its transitions the same way as done for the translation of \textsc{applied-too-long}, given this is not already done.

The difference to the previous translation is that we do not add a transition to the initial state.
Instead, we modify the outgoing transitions of $s_{ca\_cv}$ in the following way:
If previously $T(s_{ca\_cv}, f_{s}, i) = (s', \emptyset, \emptyset)$ with $f_{s}(x) = 1$ and $x\in D$, we replace $f_s$ with $f'_s$, whereby $f'_s(x \land \neg cv) = 1 \Leftrightarrow f_s(x) = 1$,
meaning the outgoing transitions can only be triggered if $cv$ does not hold.
Another difference is that here we are also interested in the transitions from $s_{ca}$ to duplicates of this state.
After all necessary duplicate states are created, these transitions are also copied and modified for the newly created states.
With \emph{newly created} we mean that if a duplicate state was already created because of the translation of an \textsc{applied-too-long} formula,
this state does not get transitions to other duplicate states even if for the same context a \textsc{stopped-too-soon} formula exist.

\begin{figure*}[htb]
	\centering
	\begin{tikzpicture}[>=stealth', scale=.8, transform shape]
		\node[draw, minimum width=1.5cm, minimum height=3.5cm] (A) {A};
		\node[circle, draw, minimum size=1cm, right=1.5cm of A] (C) {$s_{ca}$};
		
		\draw[->] ($(A.south east)+(0, .5)$) -| (C.230);
		\draw[->] ($(A.south east)+(0, .1)$) -| (C.310);
		\node[below = .5cm of C.230, left] (l1) {\(j_1\)};
		\node[below = .5cm of C.310, right] (l2) {\(j_n\)};
		\node at ($(l1)!.5!(l2)$) {\(\strut\cdots\)};
		
		\draw[->] (C.50) |- ($(A.north east)-(0, .1)$);
		\draw[->] (C.130) |- ($(A.north east)-(0, .5)$);
		\node[above = .5cm of C.130, left] (l3) {\(k_1\)};
		\node[above = .5cm of C.50, right] (l4) {\(k_m\)};
		\node at ($(l3)!.5!(l4)$) {\strut\(\cdots\)};
		\draw[thick, -Latex] (3.7, 0) -- +(1, 0);
		\begin{scope}[xshift=5.9cm]
			\node[draw, minimum width=1.5cm, minimum height=3.5cm] (A) {A};
			\node[circle, draw, minimum size=1cm, right=1.5cm of A] (C) {$s_{ca}$};
			\node[circle, draw, minimum size=1cm, right=1.5cm of C] (CC) {$s_{ca\_cv}$};
			
			\draw[->] ($(A.south east)+(0, .5)$) -| (C.230);
			\draw[->] ($(A.south east)+(0, .1)$) -| (C.310);
			\node[below = .5cm of C.230, left] (l1) {\(j_1\land\lnot cv\)};
			\node[below = .5cm of C.310, right] (l2) {\(j_n\land\lnot cv\)};
			\node at ($(l1)!.5!(l2)$) {\(\strut\cdots\)};
			
			\draw[->] ($(A.south west)+(.5, 0)$) -- +(0, -.75) -| (CC.310);
			\draw[->] ($(A.south east)-(.5, 0)$) -- +(0, -.25) -| (CC.230);
			
			\node[yshift=-.25cm, right=4.3cm of A.south, above] (l5) {\(j_1\land cv\)};
			\node[yshift=-.75cm, right=4.3cm of A.south, below] (l6) {\(j_n\land cv\)};
			\node at ($(l5)!.5!(l6)$) {\(\strut\vdots\)};
			
			\draw[->] (C.50) |- ($(A.north east)-(0, .1)$);
			\draw[->] (C.130) |- ($(A.north east)-(0, .5)$);
			\node[above = .5cm of C.130, left] (l3) {\(k_1\)};
			\node[above = .5cm of C.50, right] (l4) {\(k_m\)};
			\node at ($(l3)!.5!(l4)$) {\strut\(\cdots\)};
			\draw[->] (C) -- node[above] {\(cv\)} (CC);
			
			\draw[<-] ($(A.north west)+(.5, 0)$) -- +(0, .75) -| (CC.50);
			\draw[<-] ($(A.north east)-(.5, 0)$) -- +(0, .25) -| (CC.130);
			
			\node[yshift=.25cm, right=4.3cm of A.north, below] (l7) {$k_1\land\lnot cv$};
			\node[yshift=.75cm, right=4.3cm of A.north, above] (l8) {$k_m\land\lnot cv$};
			\node at ($(l7)!.5!(l8)$) {\(\strut\vdots\)};
		\end{scope}
	\end{tikzpicture}
	\caption{Translation of \acs{uca} type \textsc{stopped-too-soon}. $A$ is the statechart without the state $s_{ca}$.}
	\label{f:stopped-too-soon}
\end{figure*}
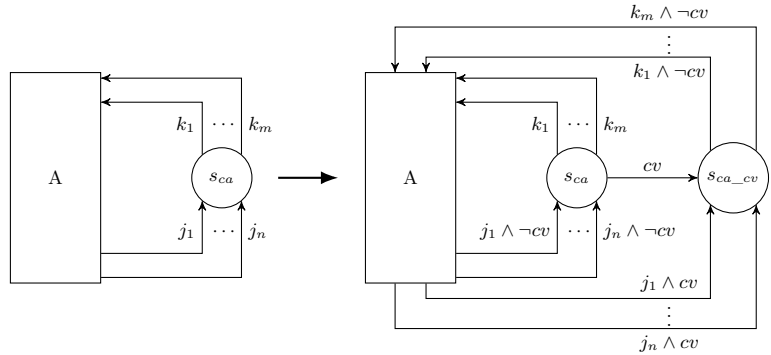

\subsection{Optimization}
\label{sec:sbm-gen-optimization}

After application of the proposed translation rules, the statechart can be further optimized.
When translating \acp{uca} of type \textsc{applied-too-long} or \textsc{stopped-too-soon} we modify the trigger for the incoming transitions to the original state.
This may lead to triggers that always evaluate to $\False$.
Transitions with such triggers can be deleted, which can lead to unreachable states that can be deleted as well.

\subsection{Construction Proof}
\label{sec:sbm-gen-proof}

We cannot translate \textsc{too-early} formulas into corresponding transitions, and thus cannot ensure that these formulas are fulfilled.
However, we can ensure other formulas are fulfilled, as stated in the following theorem.

\begin{theorem}
	For an \ac{sbm} created with the proposed construction rules, no trace exists that violates one of the \ac{ltl} formulas,
	except possibly the ones for the type \textsc{too-early}.
\end{theorem}

\begin{proof}
	Let $\varphi = \varphi_0 \land \dots \land \varphi_i$ be the \ac{ltl} formula the \ac{sbm} should fulfill, where $\varphi_0, \dots, \varphi_i$ are the formulas used to create the \ac{sbm}. 
	Let $t = t_0, t_1, \dots$ be a trace that does not fulfill $\varphi$, which means at least one subformula is not fulfilled. 
	Then there exists a $\varphi_j$ with $0\leq j \leq i, j \in \mathbb{N}$ that is not fulfilled by $t$.
	We show that the generated statechart cannot produce a trace $r = r_0, r_1, \dots$, which violates $\varphi_j$.
	$\varphi_j$ can have five possible forms, one for each \ac{uca} type except \textsc{too-early}.
	We assume the statechart can produce a trace $r$ with $r=t$.
	
	\begin{enumerate}[label=Case \arabic*:,ref=\arabic*,leftmargin=*]
		\item  \textsc{provided} (\autoref{eq:provided}) 
		
		$\varphi_j$ is of form $cv \to \neg ca$.		
		This means $t_x, t_{x+1}$ exist with $output(t_{x+1})=ca$ and $cv$ holds in reaction $x+1$.
		
		If $output(t_x)=ca$, according to the construction rules a transition from $s_{ca}$ to $s_0$ exists that triggers when $cv$ holds.
		Hence, $state(r_{x+1}) = s_0$ and $output(r_{x+1}) = \emptyset \neq \{ca\} =output(t_{x+1})$ and therefore $r\neq t$.
		
		If $output(t_x) \neq ca$, a transition from $state(r_x)$ to $state(r_{x+1})$ must exist that is triggered when $cv$ holds, and $output(r_{x+1})=ca$.
		Such a transition only exists if there exists $\varphi_k$ with $0\leq k \leq i, k \in \mathbb{N}$ of the form $cv \to ca$, which would mean $\varphi$ is not satisfiable and hence is forbidden.
		
		\item \label{case2} \textsc{not-provided} (\autoref{eq:not-provided}) 
		
		$\varphi_j$ is of form $\G((\neg cv \land \X cv) \to \X ((ca \R cv) \land \F ca))$.
		\begin{enumerate}[label=Case \ref{case2}\alph*),leftmargin=2em]
			\item A sequence $t_x, \dots, t_y$ exists with $output(t_{x+1}), \dots, output(t_{y-1}) \neq ca$, $cv$ does not hold in reaction $x$, $cv$ holds in reactions $x+1$ till $y-1$ and $cv$ does not hold in reaction $y$.
			Then either $state(r_x)= s_{ca}$ must hold or according to the construction rule a transition from $state(r_x)$ to $s_{ca}$ must exist that triggers when $cv$ holds.
			
			In the latter case, since $cv$ holds in reaction $x+1$, $state(r_{x+1}) = s_{ca}$ must hold. 
			In the first case, according to construction rules we only leave the state $ca$ if an \ac{ltl} formula exists, which states that providing $ca$ is forbidden or which states that providing another control action is necessary. 
			Both occurrences are a contradiction.
			Thus, $state(r_{x+1}) = s_{ca}$.
			
			Hence, in both cases $output(r_{x+1}) = ca \neq output(t_{x+1})$ holds.
			\item A reaction $t_x$ exists with $output(t_{z}) \neq ca$ for all $z > x$, $cv$ does not hold in reaction $x$, and $cv$ holds for every reaction after $x$.
			In this case, the same argument applies as for the other case.
			Either we are already in $s_{ca}$ and do not leave it because $cv$ holds,
			or a transition from $state(r_x)$ to $s_{ca}$ must exist that triggers when $cv$ holds.
			Hence, since $cv$ holds in reaction $x+1$, it must hold $state(r_{x+1}) = s_{ca}$. 
			Thus, $output(r_{x+1}) = ca \neq output(t_{x+1})$.
		\end{enumerate}
		\item \textsc{too-late} (\autoref{eq:too-late})
		
			$\varphi_j$ is of form $cv \to ca \land \G(\neg cv \to \X(cv \to ca))$.
			This means $t_x, t_{x+1}$ exist with $output(t_{x+1})\neq ca$,  $cv$ does not hold in reaction $x$, and $cv$ holds in reaction $x+1$.
			According to the construction rule a transition from all states to $s_{ca}$ exists that triggers when $cv$ holds.
			Since we defined the triggers of transitions uniquely to avoid non-determinism, no other transition exists that could fire.
			Hence, $state(r_{x+1}) = s_{ca}$ and thus $output(r_{x+1})=ca\neq output(t_{x+1})$.
		\item \textsc{applied-too-long} (\autoref{eq:applied-too-long})
		
			$\varphi_j$ is of form $\G((cv \land ca) \to \X (\neg cv \to \neg ca))$.
			This means a sequence $t_x, t_{x+1}$ exist with $output(t_x)=ca=output(t_{x+1})$, $cv$ holds in reaction $x$, and $cv$ does not hold in reaction $x+1$.
			According to the construction rule and since $output(t_x)=ca$, it must hold $state(r_x)=s_{ca\_cv}$.
			Since $cv$ does not hold in reaction $x+1$, the transition to $s_0$ or another state that does not send $ca$ is triggered.
			This leads to $output(r_{x+1}) \neq ca = output(t_{x+1})$.
		\item \textsc{stopped-too-soon} (\autoref{eq:stopped-too-soon})
		
			$\varphi_j$ is of form $\G ((cv \land ca) \to \X (\neg ca \to \neg cv))$.
			This means $t_x, t_{x+1}$ exist with $output(t_x)=ca$, $output(t_{x+1}) \neq ca$, and $cv$ holds in reaction $x$ and $x+1$.
			According to the construction rule and since $output(t_x)=ca$, it must hold $state(r_x)=s_{ca\_cv}$.
			Since $output(t_{x+1}) \neq ca$, an outgoing transition must have been triggered and thus $state(t_{x+1})\neq s_{ca\_cv}$.
			According to the construction, an outgoing transition can only be triggered if $cv$ does not hold, which is a contradiction to the statement that $cv$ holds in reaction $x+1$.
	\end{enumerate}
\end{proof}

\section{\aclp{dca}}
\label{sec:dca}

The behavior model that is generated with the translation presented in the last section fulfills the \ac{ltl} formulas and hence is safe.
However, a system should not only be safe but also should fulfill its system goals, 
which is typically not implied by the safety properties alone and thus not part of the proposed model generation.
However, it turns out that we can apply the machinery presented so far to achieve that aim as well.

To address the issue, we propose to extend \ac{stpa} with \emph{\acfp{dca}}.
A \ac{dca} determines in which context a control action should (not) be sent to fulfill the system goal.
In contrast to \acp{uca}, \acp{dca} have only two types: \textsc{provided} and \textsc{not-provided}.
The analysts applying \ac{stpa} already have to look at each context to determine whether \acp{uca} exist.
During that process, they can declare \acp{dca} for contexts where (not) providing a control action is desired to achieve the system goal.

These \acp{dca} can be translated to \ac{ltl} formulas just as done for the \acp{uca}.
The difference is that for \acp{dca} of type \textsc{provided} \autoref{eq:not-provided} must be used and for \acp{dca} of type \textsc{not-provided} \autoref{eq:provided}.
For a \ac{uca} of type \textsc{not-provided} the formula must ensure that the control action is provided, and for a \ac{uca} of type \textsc{provided} the formula must ensure that the control action is not provided.
Conversely, for a \ac{dca} of type \textsc{not-provided} the formula must ensure that the control action is not provided, and for a \ac{dca} of type \textsc{provided} the formula must ensure that the control action is provided.

When using these \ac{ltl} formulas that represent \acp{dca} together with the ones that prevent \acp{uca} to generate the \ac{sbm},
the resulting model will not only be safe but additionally fulfills the specified system goals.

\section{\acs{acc} Example}
\label{sec:impl}
We implemented the proposed \ac{sbm} generation in the open source tool \ac{pasta}~\cite{PetzoldKH23} to demonstrate the approach.
As an example we analyzed 
an \ac{acc} with stop and go functionality~\cite{VenhovensNA00}.
In \ac{pasta} the analyst can state abstract values for process model variables, which can be used to define \acp{uca}.
In order to infer the process model variable types, the user can define the value ranges for each abstract value.
The keywords \emph{true} and \emph{false} can be used to indicate boolean variables.
\emph{MIN} and \emph{MAX} can be used 
to state that there is no lower or upper bound respectively.
In value ranges that contain two values, `[' and `]' are used to indicate that the range should include the first or last value respectively, while `(' and `)' indicate that the value should be excluded.

Consider for example the variable \emph{speed} with the values \emph{desiredSpeed}, \emph{less\-Than\-Desired\-Speed}, and \emph{greaterThanDesiredSpeed} for the \ac{acc} as used by Abdulkhaleq \etal~\cite{AbdulkhaleqW16}.
The value ranges for these values are the following: 
\begin{align*}
	desiredSpeed &= [desiredSpeed] \\
	lessThanDesiredSpeed &= [MIN, desiredSpeed) \\
	greaterThanDesiredSpeed &= (desiredSpeed, MAX]
\end{align*}

If no value ranges are defined, we create an enum type that contains the values.
This way the context in an \ac{uca} can be translated to variables used in the behavior model.
For example a \ac{uca} with context \emph{speed=lessDesiredSpeed} is translated to \emph{speed < desiredSpeed} when generating the \ac{ltl} formula.

The \acp{dca} can be defined the same way as the \acp{uca}.
Rules can be defined for control actions and \ac{dca} types.
In the \emph{contexts} field the \acp{dca} are defined stating the context in which the control action should (not) be provided.

The user can select a controller for which the model should be generated, which triggers the generation of the \ac{sbm} as a \ac{scchart}.
An \ac{scchart} can contain several states, transitions with priorities, and input, output and internal variables \cite{vonHanxledenDM+13b}.
The resulting \ac{sbm} contains a state for each control action and an initial state as described in \autoref{sec:sbmGeneration}.
The process model variables are translated to internal variables, where the type is inferred from the value ranges.
Additionally, an enum is created that contains the control actions as values and the \texttt{controlAction} variable has the type of that enum.
The input variables are the ones that are used in the value ranges.

The \acp{uca} and \acp{dca} are translated to \ac{ltl} formulas following the proposed translation rules.
$cv$ is created by connecting the elements in the \emph{context} field of the \ac{uca}/\ac{dca} with the $\land$ operator.
For $ca$ the control action stated in the \ac{uca}/\ac{dca} is translated to the corresponding enum value $caEnum$, and 
we define $ca \Leftrightarrow controlAction = caEnum$.
The formulas are added as \emph{LTL Annotations} to the \ac{scchart}
and translated as described in \autoref{sec:sbmGeneration}.
In each created state an \emph{entry} action is defined to set \emph{controlAction} to the value represented by the state.

For the \ac{acc} example the resulting \ac{sbm} is shown in \autoref{f:acc-sbm}.
It contains four states: 
the initial one in which no control action is sent, a state in which the vehicle stops, 
a state in which the vehicle accelerates, and one in which the vehicle decelerates.
The names of the accelerating and decelerating state differ from the other two because they were generated based on a \ac{uca} with type \textsc{applied-too-long} and \textsc{stopped-too-soon} respectively.
The original states, which were just labeled \texttt{acc} and \texttt{dec} were deleted since they were not reachable.
The \texttt{stpa} file and the resulting textual \ac{scchart} are shown in the \nameref{ap:acc-example}.

\begin{figure*}[tb]
	\centering
	\includegraphics[scale=0.5]{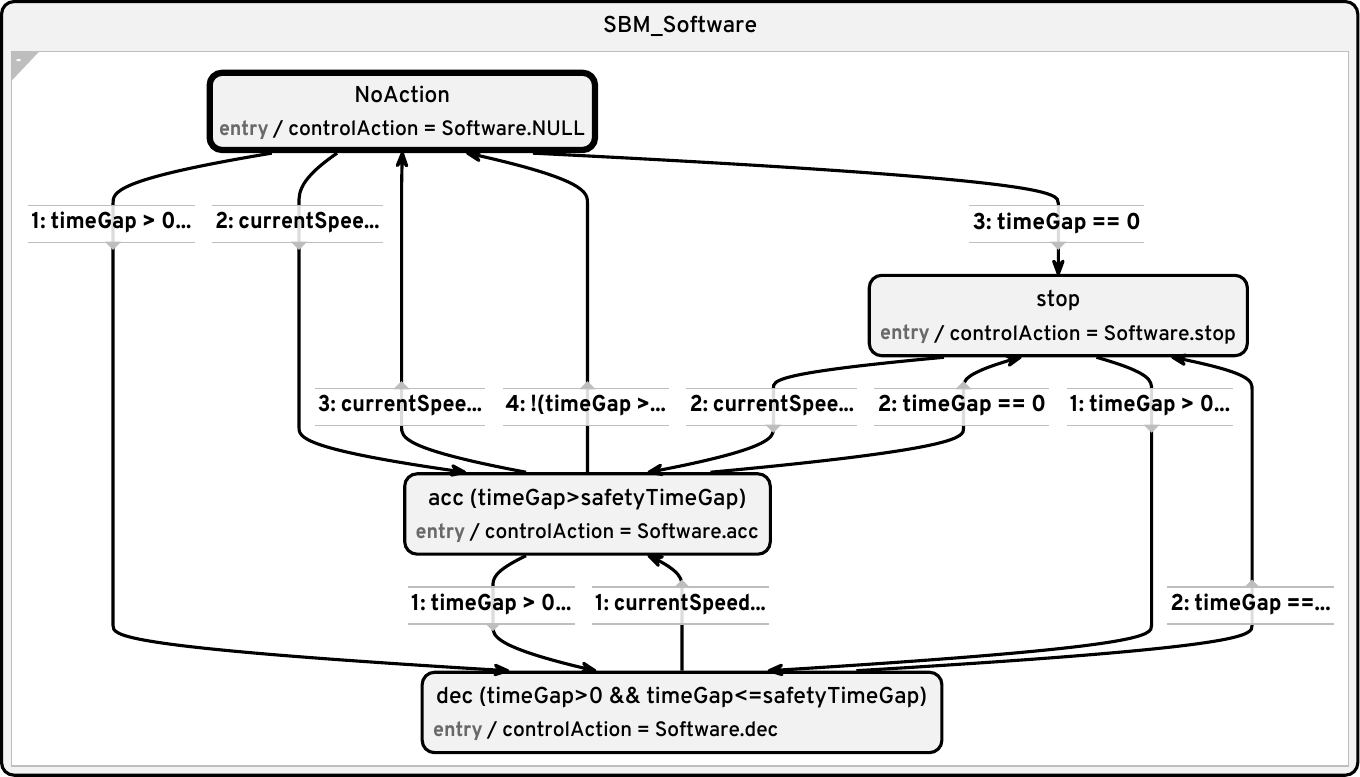}
	\caption{Automatically generated \ac{sbm} for the \ac{acc} example.}
	\label{f:acc-sbm}
\end{figure*}

\section{Discussion}
\label{sec:eva}

The resulting \ac{sbm} for the \ac{acc} example fulfills all generated \ac{ltl} formulas.
Hence, safety properties as well as liveness properties are fulfilled.
However, the proposed synthesis is not necessarily complete.
Calculations of internal values, \eg \texttt{timeGap}, cannot be automatically inferred.
The user has to modify the resulting \ac{sbm} by adding another region and stating the calculation.
This way the calculations are performed concurrently to the behavior of the system.

The same applies to the initialization of values.
Variables that are of type \emph{int} get 0 as initial value, but in the \ac{acc} example the variable \texttt{currentSpeed} should be set to some initial speed of the vehicle.
Additionally, \texttt{currentSpeed} must be updated in each reaction.
For that the user must state how the speed changes based on the sent control action or an additional input is needed, whose value is assigned to \texttt{currentSpeed}.
A complete \ac{acc} \ac{sbm} is shown in the \nameref{ap:acc-example}.

It would be possible to modify \ac{stpa} such that this missing information can be extracted automatically.
The analyst could state the effect of a control action in its definition, and the definition of the process model variables can be extended to also include initial values and calculations.
However, we think that this is out of scope for the safety analyst and should be done by a system developer.
The needed information may not be known by the analyst and does not influence the analysis.
Additionally, it would not save time to just outsource the initialization and calculation of variables to the analyst.

When the missing information is added to the \ac{sbm}, the model depicts the desired and safe behavior since the behavior is not changed by the added information.
The \ac{sbm} created manually by Abdulkhaleq \etal for the \ac{acc} is very similar to our generated one.
It contains the same four states but these are further encapsulated by a superstate.
This superstate is connected to another state that models whether the \ac{acc} is off or on.
Since the controller in the analysis has no variable for stating the status of the \ac{acc}, this is not considered in our \ac{sbm} generation.
Even if it is considered we would not generate a superstate for the four already generated states.
In future work we will work on integrating an option to define which superstates are required such that the internal behavior can be generated automatically.

In conclusion, the generated \ac{sbm} is not fully complete and can be further improved.
However, it already gives the system developers a solid foundation for an \ac{sbm}, and since risk analysis has to be done either way, the synthesis saves time.
Generating the safety properties directly based on a safety analysis could also reduce failures or missing formulas.

\section{Conclusion and Future Work}
\label{sec:conclusion}

We presented a synthesis of an \ac{sbm} from \ac{stpa}.
The first part of the synthesis translates the \acp{uca} from \ac{stpa} to \ac{ltl} formulas.
In the second part an \ac{sbm} is generated by using each control action defined in the analysis as a state and translating the \ac{ltl} formulas to corresponding transitions.
We extended \ac{stpa} by \acp{dca} to model the desired behavior as well.
They are stated just like the \acp{uca} and hence can be translated to \ac{ltl} formulas and to transitions in the same way.

The synthesis is implemented in \ac{pasta} and creates an \ac{scchart} as the \ac{sbm}.
As an example an \ac{acc} was used, which resulted in a similar \ac{sbm} as created manually by Abdulkhaleq \etal
The resulting \ac{sbm} is not complete because initialization and calculation of variables must be added manually.
However, it provides a good foundation and covers safety as well as liveness properties.

In the future we want to allow the definition of explicit state variables in the process model of a controller.
For each value of this variable a superstate can be created.
The behavior inside each superstate can be generated by 
using the presented approach with the corresponding \acp{uca}.
An open question in this approach is how to automatically generate the transitions between the superstates.
One approach could be to extend \ac{stpa} even more to define the transitions between the states, but this again may be out of scope for the safety analyst.

\bibliographystyle{splncs04}
\bibliography{cav24}

\begin{thebibliography}{10}
\providecommand{\url}[1]{\texttt{#1}}
\providecommand{\urlprefix}{URL }
\providecommand{\doi}[1]{https://doi.org/#1}

\bibitem{AbdulkhaleqW16}
Abdulkhaleq, A., Wagner, S.: {A Systematic and Semi-Automatic Safety-Based Test
  Case Generation Approach Based on Systems-Theoretic Process Analysis}. arXiv
  preprint arXiv:1612.03103  (2016)

\bibitem{BabiakKRS12}
Babiak, T., K{\v{r}}et{\'\i}nsk{\`y}, M., {\v{R}}eh{\'a}k, V., Strej{\v{c}}ek,
  J.: {LTL to B{\"u}chi Automata Translation: Fast and More Deterministic}. In:
  International Conference on Tools and Algorithms for the Construction and
  Analysis of Systems. pp. 95--109. Springer (2012)

\bibitem{ChengK93}
Cheng, K.T., Krishnakumar, A.S.: Automatic functional test generation using the
  extended finite state machine model. In: Proceedings of the 30th
  International Design Automation Conference. pp. 86--91 (1993)

\bibitem{CimattiCGR99}
Cimatti, A., Clarke, E., Giunchiglia, F., Roveri, M.: {NuSMV: A new Symbolic
  Model Verifier}. In: Computer Aided Verification: 11th International
  Conference, CAV’99 Trento, Italy, July 6--10, 1999 Proceedings 11. pp.
  495--499. Springer (1999)

\bibitem{Clarke97}
Clarke, E.M.: Model checking. In: {Foundations of Software Technology and
  Theoretical Computer Science: 17th Conference Kharagpur, India, December
  18--20, 1997 Proceedings 17}. pp. 54--56. Springer (1997)

\bibitem{DamasLL06}
Damas, C., Lambeau, B., Van~Lamsweerde, A.: {Scenarios, Goals, and State
  Machines: a Win-Win Partnership for Model Synthesis}. In: Proceedings of the
  14th ACM SIGSOFT international symposium on Foundations of software
  engineering. pp. 197--207 (2006)

\bibitem{DanieleGV99}
Daniele, M., Giunchiglia, F., Vardi, M.Y.: {Improved Automata Generation for
  Linear Temporal Logic}. In: Computer Aided Verification: 11th International
  Conference, CAV’99 Trento, Italy, July 6--10, 1999 Proceedings 11. pp.
  249--260. Springer (1999)

\bibitem{DippolitoBPU10}
D'Ippolito, N.R., Braberman, V., Piterman, N., Uchitel, S.: {Synthesis of Live
  Behaviour Models}. In: Proceedings of the eighteenth ACM SIGSOFT
  international symposium on Foundations of software engineering. pp. 77--86
  (2010)

\bibitem{GastinO01}
Gastin, P., Oddoux, D.: {Fast LTL to B{\"u}chi Automata Translation}. In:
  Computer Aided Verification: 13th International Conference, CAV 2001 Paris,
  France, July 18--22, 2001 Proceedings 13. pp. 53--65. Springer (2001)

\bibitem{GerthPVW95}
Gerth, R., Peled, D., Vardi, M.Y., Wolper, P.: {Simple On-the-fly Automatic
  Verification of Linear Temporal Logic}. In: International Conference on
  Protocol Specification, Testing and Verification. pp. 3--18. Springer (1995)

\bibitem{GiannakopoulouL02}
Giannakopoulou, D., Lerda, F.: {From states to transitions: Improving
  translation of LTL formulae to B{\"u}chi Automata}. In: International
  Conference on Formal Techniques for Networked and Distributed Systems. pp.
  308--326. Springer (2002)

\bibitem{GiannakopoulouM03}
Giannakopoulou, D., Magee, J.: {Fluent Model Checking for Event-based Systems}.
  In: Proceedings of the 9th European software engineering conference held
  jointly with 11th ACM SIGSOFT international symposium on Foundations of
  software engineering. pp. 257--266 (2003)

\bibitem{vonHanxledenDM+13b}
von Hanxleden, R., Duderstadt, B., Motika, C., Smyth, S., Mendler, M., Aguado,
  J., Mercer, S., O'Brien, O.: {SCCharts: Sequentially Constructive
  Statecharts} for safety-critical applications. Technical Report~1311,
  Christian-Albrechts-Universit{\"a}t zu Kiel, Department of Computer Science
  (Dec 2013), {ISSN 2192-6247}

\bibitem{vonHanxledenLF+22}
von Hanxleden, R., Lee, E.A., Fuhrmann, H., Schulz{-}Rosengarten, A.,
  Domr{\"{o}}s, S., Lohstroh, M., Bateni, S., Menard, C.: Pragmatics twelve
  years later: A report on {Lingua Franca}. In: 11th International Symposium on
  Leveraging Applications of Formal Methods, Verification and Validation
  (ISoLA). Lecture Notes in Computer Science, vol. 13702, pp. 60--89. Springer,
  Rhodes, Greece (Oct 2022). \doi{10.1007/978-3-031-19756-7_5}

\bibitem{Harel87}
Harel, D.: Statecharts: {A} visual formalism for complex systems. Science of
  Computer Programming  \textbf{8}(3),  231--274 (Jun 1987)

\bibitem{KruegerGSB98m}
Kr{\"u}ger, I., Grosu, R., Scholz, P., Broy, M.: {From MSCs to Statecharts}.
  In: IFIP Working Conference on Distributed and Parallel Embedded Systems. pp.
  61--71. Springer (1998)

\bibitem{LeeS17}
Lee, E.A., Seshia, S.A.: Introduction to Embedded Systems, A Cyber-Physical
  Systems Approach, Second Edition. MIT Press (2017),
  \url{http://LeeSeshia.org}

\bibitem{LevesonT18}
Leveson, N., Thomas, J.P.: {STPA} {Handbook}. MIT Partnership for Systems
  Approaches to Safety and Security (PSASS)  (2018),
  \url{http://psas.scripts.mit.edu/home/get_file.php?name=STPA_handbook.pdf}

\bibitem{Leveson16}
Leveson, N.G.: Engineering a {Safer} {World}: {Systems} {Thinking} {Applied} to
  {Safety}. The MIT Press (2016)

\bibitem{PetzoldKH23}
Petzold, J., Krei{\ss}, J., von Hanxleden, R.: {PASTA: Pragmatic Automated
  System-Theoretic Process Analysis}. In: 53rd Annual {IEEE/IFIP} International
  Conference on Dependable Systems and Network, {DSN} 2023, Porto, Portugal,
  June 27-30, 2023. pp. 559--567. {IEEE} (2023).
  \doi{10.1109/DSN58367.2023.00058}

\bibitem{SomenziB00}
Somenzi, F., Bloem, R.: {Efficient B{\"u}chi Automata from LTL Formulae}. In:
  Computer Aided Verification: 12th International Conference, CAV 2000,
  Chicago, IL, USA, July 15-19, 2000. Proceedings 12. pp. 248--263. Springer
  (2000)

\bibitem{Thomas13}
Thomas, J.P.: Extending and {A}utomating a {S}ystems-{T}heoretic {H}azard
  {A}nalysis for {R}equirements {G}eneration and {A}nalysis. Ph.D. thesis,
  Massachusetts Institute of Technology (2013)

\bibitem{UchitelBC08}
Uchitel, S., Brunet, G., Chechik, M.: {Synthesis of Partial Behavior Models
  from Properties and Scenarios}. IEEE Transactions on Software Engineering
  \textbf{35}(3),  384--406 (2008)

\bibitem{UchitelKM01}
Uchitel, S., Kramer, J., Magee, J.: {Detecting Implied Scenarios in Message
  Sequence Chart Specifications}. ACM SIGSOFT Software Engineering Notes
  \textbf{26}(5),  74--82 (2001)

\bibitem{VenhovensNA00}
Venhovens, P., Naab, K., Adiprasito, B.: {Stop and Go Cruise Control}. Tech.
  rep., SAE Technical Paper (2000)

\end{thebibliography}

\newpage

\appendix

\section*{Appendix}
\label{ap:acc-example}

The analysis we have done for the example \ac{acc} system is shown in \autoref{code:stpa-acc}.
It is an \texttt{stpa} file created in \ac{pasta}.
Executing the generation of the \ac{sbm} for this file results in the \ac{scchart} file shown in \autoref{code:scchart-acc}.
The textual \ac{scchart} is annotated with the \ac{ltl} formula we generated and used for the synthesis of the model.
Generally, \acp{scchart} are realized with a text-first approach \cite{vonHanxledenLF+22}.
This means, the \ac{scchart} is defined textually in an editor and the corresponding visualization is generated automatically.

\begin{lstlisting}[language=stpa, caption=\ac{acc} example analysis in \ac{pasta}., label=code:stpa-acc]
Losses
L1 "The ACC robot crashes the robot ahead"

Hazards
H1 "The ACC software does not keep a safe distance from the vehicle robot ahead." [L1]

ControlStructure
ACC {
    Software {
        processModel {
            currentSpeed: [desiredSpeed=[desiredSpeed], lessDesiredSpeed = [MIN, desiredSpeed), greaterDesiredSpeed=(desiredSpeed, MAX]]
            timeGap: [lessSafetyTimeGap=(0, safetyTimeGap], greaterSafetyTimeGap=(safetyTimeGap, MAX], zero=[0]]
        }
        controlActions {
            [acc "Accelerate", dec "Decelerate", stop "Fully Stop"] -> Robot
        }
    }
    Robot {
        feedback {
            [speed "Speed"] -> Software
        }
    }
}

ContextTable
// acceleration
RL1 {
    controlAction: Software.acc
    type: provided
    contexts: {
        UCA1 [timeGap = lessSafetyTimeGap] [H1]
    }
}
RL2 {
    controlAction: Software.acc
    type: applied-too-long
    contexts: {
        UCA2 [timeGap = greaterSafetyTimeGap] [H1]
    }
}
RL3 {
    controlAction: Software.acc
    type: provided
    contexts: {
        UCA3 [timeGap = zero] [H1]
    }
}
// deceleration
RL4 {
    controlAction: Software.dec
    type: too-late
    contexts: {
        UCA4 [timeGap = lessSafetyTimeGap] [H1]
    }
}
RL5 {
    controlAction: Software.dec
    type: stopped-too-soon
    contexts: {
        UCA5 [timeGap = lessSafetyTimeGap] [H1]
    }
}
RL6 {
    controlAction: Software.dec
    type: not-provided
    contexts: {
        UCA6 [timeGap = lessSafetyTimeGap] [H1]
    }
}
// stop
RL7 {
    controlAction: Software.stop
    type: too-late
    contexts: {
        UCA7 [timeGap = zero] [H1]
    }
}
RL8 {
    controlAction: Software.stop
    type: not-provided
    contexts: {
        UCA8 [timeGap = zero] [H1]
    }
}

DCAs
R1 {
    controlAction: Software.acc
    type: provided
    contexts: {
        DCA1 [currentSpeed = lessDesiredSpeed, timeGap = greaterSafetyTimeGap]
    }
}
R2 {
    controlAction: Software.acc
    type: not-provided
    contexts: {
        DCA2 [currentSpeed = greaterDesiredSpeed]
    }
}
\end{lstlisting}

\begin{lstlisting}[language=sctx, caption=Automatically generated textual \ac{scchart} based on the \ac{acc} analysis. The \ac{ltl} formulas are here highlighted as done in \autoref{f:formulas-summary} with the implicant in green and the implicand in orange., label=code:scchart-acc]
@LTL G (`(timeGap > 0 && timeGap <= safetyTimeGap)` -> ~(controlAction != Software.acc)~), "UCA1" 
@LTL G (`(timeGap > safetyTimeGap && controlAction == Software.acc)` -> ~(X((!(timeGap > safetyTimeGap)) -> controlAction != Software.acc))~), "UCA2" 
@LTL G (`(timeGap == 0)` -> ~(controlAction != Software.acc)~), "UCA3" 
@LTL ((timeGap > 0 && timeGap <= safetyTimeGap) -> (controlAction == Software.dec)) && G (`(!(timeGap > 0 && timeGap <= safetyTimeGap))` -> ~(X((timeGap > 0 && timeGap <= safetyTimeGap) -> (controlAction == Software.dec)))~), "UCA4" 
@LTL G (`(timeGap > 0 && timeGap <= safetyTimeGap && controlAction == Software.dec)` -> ~(X((controlAction != Software.dec) -> (!(timeGap > 0 && timeGap <= safetyTimeGap))))~), "UCA5" 
@LTL G (`(timeGap > 0 && timeGap <= safetyTimeGap)` -> ~(controlAction == Software.dec)~), "UCA6" 
@LTL ((timeGap == 0) -> (controlAction == Software.stop)) && G (`(!(timeGap == 0))` -> ~(X((timeGap == 0) -> (controlAction == Software.stop)))~), "UCA7" 
@LTL G (`(timeGap == 0)` -> ~(controlAction == Software.stop)~), "UCA8" 
@LTL G (`(currentSpeed > desiredSpeed)` -> ~(controlAction != Software.acc)~), "DCA1" 
@LTL G (`(currentSpeed < desiredSpeed && timeGap > safetyTimeGap)` -> ~(controlAction == Software.acc)~), "DCA2" 
scchart SBM_Software {

enum Software{acc, dec, stop, NULL}
int timeGap
input int safetyTimeGap
int currentSpeed
input int desiredSpeed
ref Software controlAction

initial state NoAction {
    entry do controlAction = Software.NULL
}
if timeGap > 0 && timeGap <= safetyTimeGap go to dec_timeGapGreaterThan0AndtimeGapLessOrEqualTosafetyTimeGap
if currentSpeed < desiredSpeed && timeGap > safetyTimeGap go to acc_timeGapGreaterThansafetyTimeGap
if timeGap == 0 go to stop

state stop {
    entry do controlAction = Software.stop
}
if timeGap > 0 && timeGap <= safetyTimeGap go to dec_timeGapGreaterThan0AndtimeGapLessOrEqualTosafetyTimeGap
if currentSpeed < desiredSpeed && timeGap > safetyTimeGap go to acc_timeGapGreaterThansafetyTimeGap

state acc_timeGapGreaterThansafetyTimeGap "acc (timeGap > safetyTimeGap)" {
    entry do controlAction = Software.acc
}
if timeGap > 0 && timeGap <= safetyTimeGap go to dec_timeGapGreaterThan0AndtimeGapLessOrEqualTosafetyTimeGap
if timeGap == 0 go to stop
if currentSpeed > desiredSpeed go to NoAction
if !(timeGap > safetyTimeGap) go to NoAction

state dec_timeGapGreaterThan0AndtimeGapLessOrEqualTosafetyTimeGap "dec (timeGap > 0 && timeGap <= safetyTimeGap)" {
    entry do controlAction = Software.dec
}
if currentSpeed < desiredSpeed && timeGap > safetyTimeGap && !(timeGap > 0 && timeGap <= safetyTimeGap) go to acc_timeGapGreaterThansafetyTimeGap
if timeGap == 0 && !(timeGap > 0 && timeGap <= safetyTimeGap) go to stop
}
\end{lstlisting}

\autoref{f:acc-sbm-whole-labels} shows the visualization of the generated textual \ac{scchart}.
As mentioned in \autoref{sec:eva} this model is not complete yet.
We have to define the calculation of \texttt{timeGap} and the effect of each control action manually.
\autoref{f:acc-sbm-complete} shows an \ac{scchart} where we have added these missing information.
We declared \texttt{during} actions in each state to model the effect of the corresponding control action
and defined an additional region in which \texttt{timeGap} is calculated.

\begin{figure*}[ptb]
	\centering
	\includegraphics[scale=0.483]{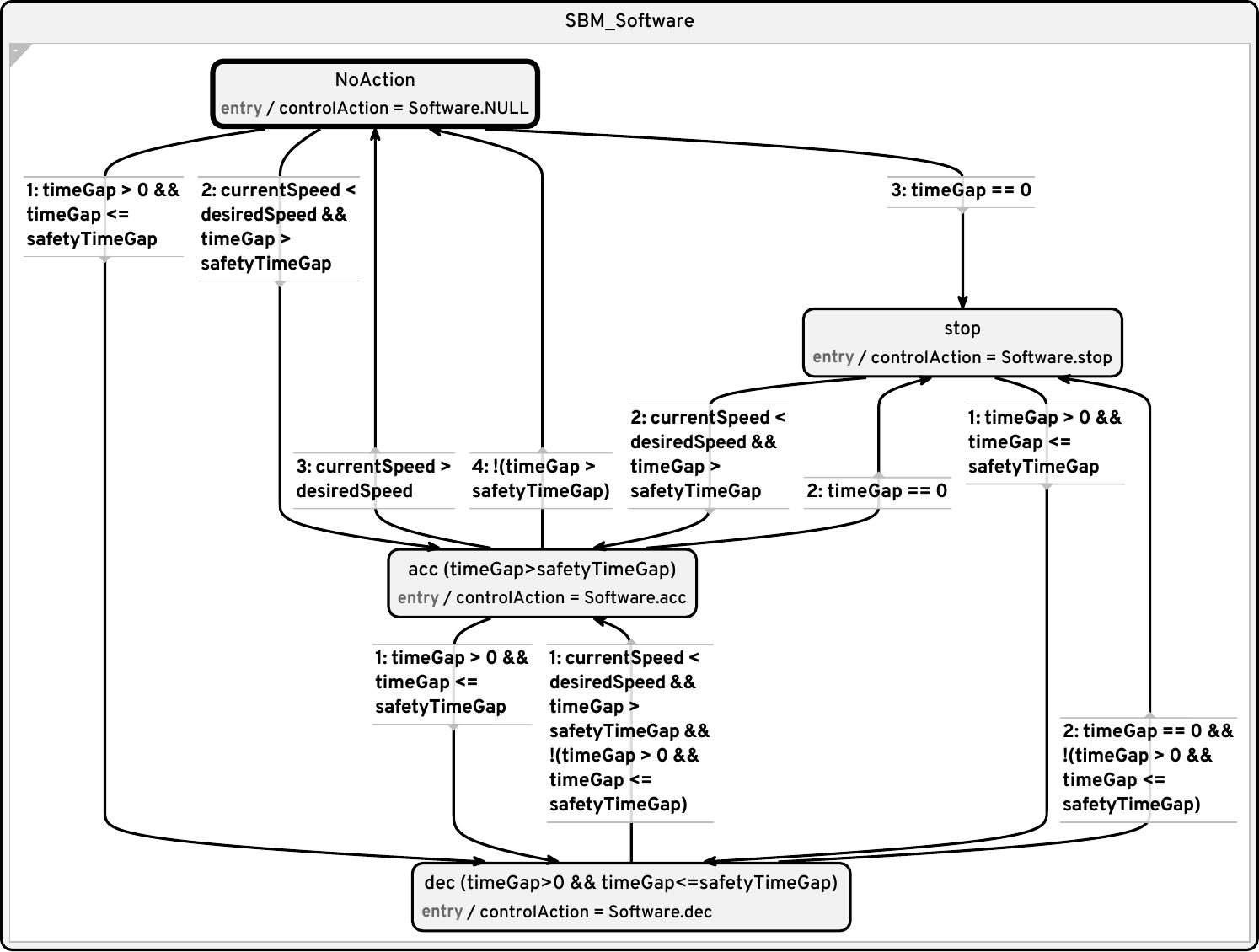}
	\caption{The \ac{acc} \ac{sbm} as an \ac{scchart} with complete transition labels.}
	\label{f:acc-sbm-whole-labels}
\end{figure*}

\begin{figure*}[ptb]
	\centering
	\includegraphics[scale=0.49]{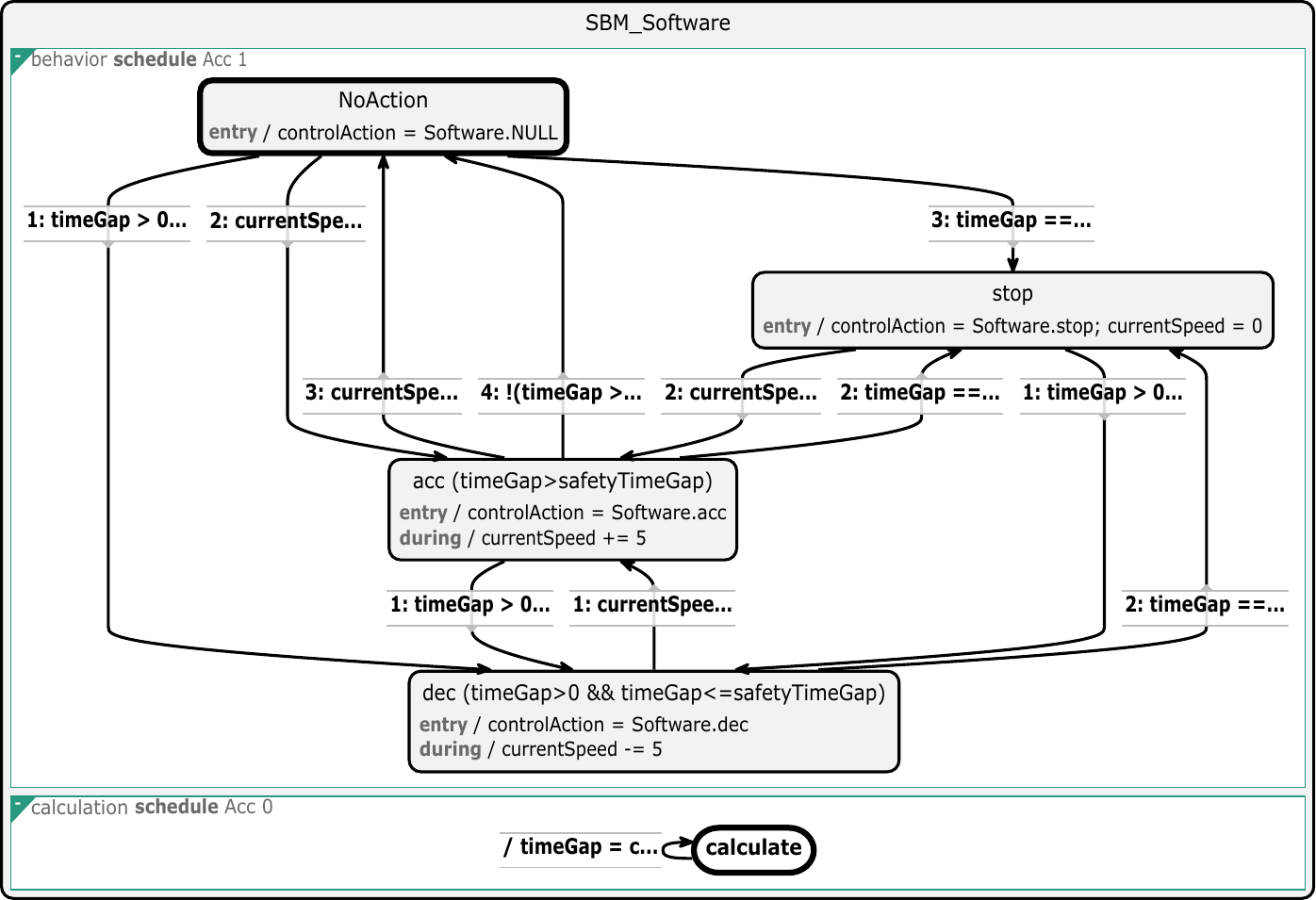}
	\caption{The \ac{acc} \ac{sbm} containing the missing information.}
	\label{f:acc-sbm-complete}
\end{figure*}

\end{document}